\documentclass[letterpaper,journal]{IEEEtran}
\usepackage{amsfonts,amsmath,amssymb}
\usepackage{bm}
\usepackage{booktabs}
\usepackage{calligra}
\usepackage{cite}
\usepackage{eucal}
\usepackage{float}
\usepackage{graphicx}
\usepackage{longtable}
\usepackage{mathtools}
\usepackage{multirow}
\usepackage[inter-unit-product =\cdot]{siunitx}
\usepackage{tabu}
\usepackage{textcomp}
\usepackage{tikz}
\usepackage{url}
\usepackage{savesym}
\savesymbol{theoremstyle}
\usepackage{amsthm}
\usetikzlibrary{calc,patterns,angles,quotes}

\newtheorem{theorem}{Theorem}[section]

\theoremstyle{remark}
\newcommand{\mat}[1]{\bm{#1}}
\renewcommand{\vec}[1]{\mathbf{#1}}

\DeclareFontFamily{U}{mathx}{\hyphenchar\font45}
\DeclareFontShape{U}{mathx}{m}{n}{<-> mathx10}{}
\DeclareSymbolFont{mathx}{U}{mathx}{m}{n}
\DeclareMathAccent{\widebar}{0}{mathx}{"73}
\DeclareMathAlphabet{\mathcalvar}{OMS}{cmsy}{m}{n}
\DeclareMathOperator{\Ero}{Ero}
\newcommand{\Mod}[1]{\ (\mathrm{mod}\ #1)}
\let\emptyset\varnothing

\newcount\Comments  
\newcommand{\kibitz}[2]{\ifnum\Comments=0\textcolor{#1}{#2}\fi}
\newcommand{\Fangyu}[1]{\kibitz{magenta} {/* Fangyu: #1 */}}

\begin{document}

\title{Composing MPC with LQR and Neural Network\\for Amortized Efficiency and Stable Control}

\author{Fangyu Wu, Guanhua Wang, Siyuan Zhuang, Kehan Wang,\\ Alexander Keimer, Ion Stoica, and Alexandre Bayen
\thanks{The article is submitted to the IEEE Transactions on Automation Sciences and Engineering on \today~for peer review.
This material is based upon work supported by the National Science Foundation under Grant Numbers CNS-1837244.
Any opinions, findings, and conclusions or recommendations expressed in this material are those of the author(s) and do not necessarily reflect the views of the National Science Foundation.
This material is based upon work supported by the U.S. Department of Energy’s Office of Energy Efficiency and Renewable Energy (EERE) award number CID DE-EE0008872. 
The views expressed herein do not necessarily represent the views of the U.S. Department of Energy or the United States Government.
}
\thanks{All authors are with the Department of Electrical Engineering and Computer Sciences at the University of California, Berkeley, Berkeley, CA 94709 USA. 
For correspondence, please contact Fangyu Wu (e-mail: fangyuwu@berkeley.edu).}
}


\IEEEpubid{
\copyright~2022 IEEE}

\maketitle

\begin{abstract}
Model predictive control (MPC) is a powerful control method that handles dynamical systems with constraints.
However, solving MPC iteratively in real time, i.e., implicit MPC, remains a computational challenge.
To address this, common solutions include explicit MPC and function approximation.
Both methods, whenever applicable, may improve the computational efficiency of the implicit MPC by several orders of magnitude. 
Nevertheless, explicit MPC often requires expensive pre-computation and does not easily apply to higher-dimensional problems.
Meanwhile, function approximation, although scales better with dimension, still requires pre-training on a large dataset and generally cannot guarantee to find an accurate surrogate policy, the failure of which often leads to closed-loop instability.
To address these issues, we propose a triple-mode hybrid control scheme, named Memory-Augmented MPC, by combining a linear quadratic regulator, a neural network, and an MPC.
From its standard form, we derive two variants of such hybrid control scheme: one customized for chaotic systems and the other for slow systems.
The proposed scheme does not require pre-computation and is capable of improving the \textit{amortized} running time of the composed MPC with a well-trained neural network.
In addition, the scheme maintains closed-loop stability with \textit{any} neural networks of proper input and output dimensions, alleviating the need for certifying optimality of the neural network in safety-critical applications.
\end{abstract}

\def\abstractname{Note to Practitioners}
\begin{abstract}
This article was motivated by the need to reduce the amortized cost of MPC in repetitive industrial robotic applications, where long-term operational cost is important and safety is critical.
Examples of such applications include factory robotic arm manipulation and fixed-route quadcopter payload transport.
Unlike explicit MPC or function approximation, our approach does not require any pre-computation or pre-training.
Rather, it attains task proficiency over time by learning a surrogate neural network on the spot and by gradually replacing the costly MPC with the more efficient surrogate model so long as safety permits.
Consequently, the proposed scheme incurs a learning cost during the initial phase of the deployment but usually becomes more adept on the task afterwards, leading to amortized efficiency.
\end{abstract}


\section{Introduction}
\label{sec:introduction}
\IEEEPARstart{M}odel predictive control (MPC) is a powerful method for controlling dynamical systems with constraints~\cite{rawlings2009model,borrelli2017predictive}.
It has been widely used in robotics, with applications ranging from ground~\cite{borrelli2005mpc} and aerial~\cite{alexis2012model} vehicle maneuvers to humanoid~\cite{kuindersma2016optimization} and quadruped~\cite{bledt2018cheetah} robot control.

When it comes to \textit{implement} MPC, one has two conventional approaches:
\textsl{1)} implicit MPC, which solves an optimization problem in real time with a preferably efficient problem formulation and numerical scheme,
\textsl{2)} explicit MPC, which computes every solution of the optimization problem via parametric optimization \textit{offline} and looks up the computed solutions \textit{online}.

Explicit MPC often has faster running time than implicit MPC but requires expensive pre-computation.
For problems of many variables and of intricate constraints, such pre-computation tends to demand a prohibitively large amount of time and memory.

Alternatively, one can approximate an implicit MPC controller with a neural network (NN) through supervised learning.
Neural network models are more flexible when it comes to pre-training and generally scales better to higher-dimensional problems.
Nevertheless, to our best knowledge, the NN function approximation approach has yet to fully address the following open problems:
\textsl{1)} it is difficult, if not impossible, to guarantee that the approximation will always converge to a solution of bounded approximation errors;
\textsl{2)} it is challenging to establish closed-loop stability without any knowledge of the accuracy of the approximation.
\IEEEpubidadjcol

Motivated by the above shortcomings of explicit MPC and NN function approximation methods and inspired by~\cite{michalska1993robust}, we propose a triple-mode hybrid control scheme, named \textit{Memory-Augmented MPC} (MAMPC).
The core idea is to mix a costly MPC controller with an efficient linear quadratic regulator (LQR) controller and an efficient NN controller, whenever stability permits.
Unlike conventional explicit MPC or function approximation methods, our method does not require any form of pre-computing.
Instead, it is operational on day one, in spite of being inefficient, and learns to be more proficient on the spot.

We summarize the main contributions of our work below.
\begin{itemize}
    \item We present a novel controller design, i.e., Memory-Augmented MPC, in its standard form and two modified forms.
    \item We prove stability of Memory-Augmented MPC without imposing any condition on the approximation errors of the neural network.
    \item We demonstrate amortized computational efficiency of Memory-Augmented MPC via four numerical experiments.
\end{itemize}

\section{Related Works}
\label{sec:related-works}

A major limitation of MPC has been a lack of computational efficiency as evident in applications such as the Atlas humanoid robot~\cite{kuindersma2016optimization} and the MIT Cheetah robot~\cite{bledt2018cheetah}.
To overcome this shortcoming, one typically has three approaches: \textsl{1)} developing an efficient implicit MPC policy, \textsl{2)} pre-compiling an explicit MPC policy, and \textsl{3)} approximating a slow implicit MPC policy by an efficient surrogate policy such as a NN.

Implicit MPC approach implements the MPC policy by devising an efficient numerical algorithm for solving the optimization program.
This is perhaps the most common approach of implementing an MPC policy.
Efficient implicit MPC controller often depends on custom optimization solvers that leverage structures specific to the problems they solve, such as~\cite{diehl2005real, mastalli2020crocoddyl}.
Implementation of such implicit MPC is usually done in a fast low-level programming language such as C.
For example, see~\cite{houska2011auto,mattingley2012cvxgen}.
Moreover, recent advances in software and hardware enable implicit MPC controllers to be deployed in traditionally challenging problems, including~\cite{faulwasser2016implementation,kleff2021high}.

When existing implicit MPC methods do not provide satisfactory latency, one sometimes resorts to explicit MPC, whenever such alternative is applicable.
Explicit MPC approach implements the MPC policy by pre-computing the solution of the optimization problem and looking up the solution in real time.
At its core, explicit methods achieve reduced running time at the cost of expensive pre-computation and increased memory footprint.
To reduce memory footprint, numerous works have been developed~\cite{summers2011multiresolution,kvasnica2013complexity}.
Nevertheless, due to poor scalability of the pre-computation step, explicit MPC has more limited use cases than implicit MPC.

When the explicit MPC method fails to apply to a problem, one may consider another alternative, namely, function approximation.
Function approximation approach implements an MPC policy by approximating the costly MPC control law with an efficient surrogate model, such as a NN, and uses that surrogate model for fast online deployment.
Early works in this direction include \cite{parisini1995receding,parisini1998nonlinear}, which has described how to design a surrogate NN controller and has proposed conditions to guarantee its closed-loop stability.
Later development has focused on guaranteeing different theoretical properties of the NN-controlled closed-loop system. 
For example, \cite{zhang2019safe} proposes a NN method with \textit{probabilistic} optimality bounds;
\cite{paulson2020approximate} has designed a projection operator that can modify a NN controller to a stabilizing one for linear systems;
and \cite{karg2020stability} proposes a mixed-integer linear program to certify stability and feasibility of ReLU-activated NN controllers.
Besides, rather than replacing the implicit MPC controller,~\cite{chen2022large} uses a NN to warm start an optimization solver, which may then speed up convergence of the implicit MPC solver.
For a comprehensive study, readers may consult~\cite{zoppoli2020neural}.
Applications of function approximation methods have found many successes in practice, such as~\cite{nubert2020safe}.

It is worth noting that closed-loop stability of the function approximation methods generally requires reasonable convergence of NN model to the original implicit MPC control law.
Nevertheless, training a NN is commonly done through a black-box solver, such as~\cite{DBLP:journals/corr/KingmaB14}, the convergence of which is generally not guaranteed.
In contrary, we propose a novel function approximation approach named Memory-Augmented MPC, which is always stable and always guarantees constraint satisfaction, without imposing \textit{any} condition on the convergence and optimality of the NN.
The accuracy of the NN only affects running time: the more accurate the NN is, the more efficient the overall control scheme is.
As a result, it can be deployed without any training, with the expectation that it will attain computational efficiency over time via learning.

As a final remark, a quick comparison of MAMPC with a few key existing works is in order.
We acknowledge that our method is inspired by the seminal early work in dual-mode MPC~\cite{michalska1993robust}, which combines an MPC with a LQR.
An advantage of our approach over~\cite{michalska1993robust} is that it could potentially provide more reduction in latency with the help of a fast intermediate NN mode.
As detailed in Section~\ref{sec:discussions}, with parallelization, MAMPC is \textit{at worst} equivalent to~\cite{michalska1993robust}.
Moreover, compared to the techniques in~\cite{paulson2020approximate,chen2022large}, which only apply to linear plants with convex quadratic objectives, our method works for any nonlinear problems.


\section{Composing MPC, LQR, and NN}
\label{sec:composing-mpc-lqr-and-nn}

We begin the section by first introducing some basic notations, terminologies, and a key stability theorem of MPC.
Next, we present the standard construction in MAMPC, followed by two variations of such construction, namely, alternating-authority MAMPC and way-point MAMPC.

\subsection{Basics of MPC}
\label{subsec:basics-of-mpc}
Consider the discrete-time time-invariant constrained dynamical systems of the form
\begin{equation}
  \vec{x}[i+1] = f(\vec{x}[i], \vec{u}[i]), \quad i \in \mathbb{Z}_{\geq 0}
  \label{eq:sys}
\end{equation}
where
\begin{equation*}
    \begin{aligned}
      &\vec{x}[0] = \vec{x}_{0} \quad \text{ for some }\vec{x}_{0} \in \mathbb{R}^{n},\\
      &(\vec{x}[i], \vec{u}[i]) \in \mathbb{A}, \quad \forall i > 0,
    \end{aligned}
\end{equation*}
where 
$\mathbb{A} \subset \mathbb{R}^{n} \times \mathbb{R}^{m}$ is a closed set defining the system constraints,
$\vec{x}[i]$ is the state at time index $i$,
$\vec{u}[i]$ is the input at time index $i$,
$f: \mathbb{R}^{n} \times \mathbb{R}^{m} \rightarrow \mathbb{R}^{n}$ is continuously differentiable in both its components, with an equilibrium point at ($\vec{0}_{n}$, $\vec{0}_{m}$), i.e., $f(\vec{0}_{n}, \vec{0}_{m}) = \vec{0}_{n}$, 
and the linearized time-invariant system of $f(\cdot, \cdot)$ around the equilibrium is stabilizable.

Let
$\mathbb{X}$
be state constraint set and 
$\mathbb{U}$
be input constraint set.
We consider system constraints of the following form $\mathbb{A} \coloneqq \mathbb{X} \times \mathbb{U}$.
The system dynamics is a map $f:\mathbb{A} \rightarrow \mathbb{X}$ and any control law is a map $u: \mathbb{X} \rightarrow \mathbb{U}$.

Consider a finite planning and control horizon $N$.
Let $\mathcal{X}^{1:N} \coloneqq (\vec{x}[1], \dots, \vec{x}[N])$ be a trajectory of the system and $\mathcal{U}^{0:N-1} \coloneqq (\vec{u}[0], \dots, \vec{u}[N-1])$ be the corresponding ordered control sequence.

Provided with the dynamical system~\eqref{eq:sys}, an MPC is an optimal control law implicitly defined through the following set of optimization problems
\begin{equation}
    \begin{aligned}
        \min_{\mathcal{X}_{i}^{1:N}, \mathcal{U}_{i}^{0:N-1}} \quad & \sum_{k=0}^{N-1}{c(\vec{x}[k|i], \vec{u}[k|i])} + c_{f}(\vec{x}[N|i])\\
        \textrm{s.t.:} \quad
        &\vec{x}[0|i] = \vec{x}[i],\\
        & \vec{x}[k+1|i] = f(\vec{x}[k|i], \vec{u}[k|i]),\\
        & (\vec{x}[k|i], \vec{u}[k|i]) \in \mathbb{A},\\
        & \vec{x}[N|i] \in \mathbb{X}_f,\\
        & k = 0, \dots, N-1,
    \end{aligned}
    \label{eq:mpc}
\end{equation}
where $N$ is planning and control horizon,
$c: \mathbb{A} \rightarrow \mathbb{R}_{+}$ is a continuous stage cost function,
$c_{f}: \mathbb{X}_{f} \rightarrow \mathbb{R}_{+}$ is a continuous terminal cost function,
$\vec{x}[k|i]$ is the predicted state $k$ steps ahead of \textit{present} state $\vec{x}[i]$,
$\vec{u}[k|i]$ is the anticipated input that generates $\vec{x}[k+1|i]$ from $\vec{x}[k|i]$,
$\mathcal{X}_{i}^{1:N} \coloneqq (\vec{x}[1|i], \dots, \vec{x}[N|i])$ and 
$\mathcal{U}_{i}^{0:N-1} \coloneqq (\vec{u}[0|i], \dots, \vec{u}[N-1|i])$ are two sets of optimization variables corresponding to predicted states and anticipated inputs,
$\mathbb{X}_{f}$ is a terminal constraint set containing the origin usually designed to guarantee asymptotic stability.
If optimization problem \eqref{eq:mpc} is feasible and admits an optimal solution $(\mathcal{X}_{i}^{1:N*}, \mathcal{U}_{i}^{0:N-1*})$, then the MPC control law selects the first element of $\mathcal{U}_{i}^{0:N-1*}$ as the input at time index $i$, that is,
\begin{equation}
   u_{\textnormal{MPC}}(\vec{x}[i]) \coloneqq \vec{u}^{*}[0|i] = \mathcal{U}_{i}^{0:N-1*}[0].
\end{equation}
This process is repeated at the next time index $i+1$, until some termination criterion is met.
As a result, we obtain a closed-loop system
\begin{equation}
    \vec{x}[i+1] = f(\vec{x}[i], u_{\textnormal{MPC}}(\vec{x}[i])), \quad i \in \mathbb{Z}_{\geq 0},
    \label{eq:clsys-mpc}
\end{equation}
with initial condition
$\vec{x}[0] = \vec{x}_{0} \in \mathbb{X}_{0} \coloneqq \{\vec{x}_{0} \in \mathbb{R}^{n} \mid \text{problem \eqref{eq:mpc} is feasible}\}$
is the set of all admissible initial states.
For well-posedness of MPC, we require $ \mathbb{X} \subseteq \mathbb{X}_{0}$.

Let 
$J(\vec{x}[i]) \coloneqq \sum_{k=0}^{N-1}{c(\vec{x}[k|i], \vec{u}[k|i])} + c_{f}(\vec{x}[N|i])$ 
be the cumulative cost and $J^{*}(\vec{x}[i]) \coloneqq \sum_{k=0}^{N-1}{c(\vec{x}^{*}[k|i], \vec{u}^{*}[k|i])} + c_{f}(\vec{x}^{*}[N|i])$ be the optimal cumulative cost.
By Lyapunov stability theorem, we can derive the following sufficient condition for asymptotic stability of MPC as in Theorem 12.2 in~\cite{borrelli2017predictive}.

\begin{theorem}[Local Asymptotic Stability of MPC]
  Assume that
  \begin{itemize}
    \item The stage cost and terminal cost are continuous and positive definite.
    \item $\mathbb{X}_{f}$ is control invariant.
    \item For all $\vec{x} \in \mathbb{X}_{f}$,
    \begin{align*}
        \min_{\vec{u}: \vec{u} \in \mathbb{U}, f(\vec{x}, \vec{u}) \in \mathbb{X}_{f}} \quad c(\vec{x}, \vec{u}) - c_{f}(\vec{x}) + c_{f}(f(\vec{x}, \vec{u})) \leq 0.
    \end{align*}
  \end{itemize}
  then the MPC problem~\eqref{eq:mpc} is persistently feasible and the closed-loop system~\eqref{eq:clsys-mpc} is locally asymptotically stable with respect to the origin.
  In addition, $J^{*}(\cdot)$ is a Lyapunov function and $\mathbb{X}_{0}$ is a positively invariant region of attraction.
  \label{thm:as-mpc}
\end{theorem}

With the above definitions and theorem, we are ready to describe the composition rules of MPC, LQR, and NN.

\subsection{Composition Rules}
\label{subsec:composition-rules}
Provided with a computationally costly implicit MPC $u_{\textnormal{MPC}}$ that satisfies Theorem~\eqref{thm:as-mpc}, we propose a hybrid control scheme, namely, \textit{Memory-Augmented Model Predictive Control} (MAMPC), $u_{\textnormal{MAMPC}} : \mathbb{X} \rightarrow \mathbb{U}$, by augmenting $u_{\textnormal{MPC}}$ with a LQR controller $u_{\textnormal{LQR}}$ and a NN controller $u_{\textnormal{NN}}$, where the LQR controller and the NN controller are independently developed from the implicit MPC.
We first describe how the LQR and NN controllers are derived from the implicit MPC and then present ways to combine these controllers.

\subsubsection{LQR from MPC}
\label{subsubsec:lqr-from-mpc}
Provided with an MPC, we derive an infinite-time LQR controller by \textsl{1)} linearizing the system dynamics around the equilibrium, \textsl{2)} removing stage constraints, terminal constraint, and terminal cost, \textsl{3)} taking planning and control horizon $N$ to $\infty$, and \textsl{4)} replacing stage cost $c(\cdot)$ with a positive definite quadratic cost, if it is not already so in the original MPC.

Formally, the MPC-induced LQR problem is defined as follows
\begin{equation}
    \begin{aligned}
        \min_{\mathcal{X}_{i}^{1:\infty}, \mathcal{U}_{i}^{0:\infty}} \quad & \sum_{k=0}^{\infty} \vec{x}^{\top}[k|i] \mat{Q} \vec{x}[k|i] + \vec{u}^{\top}[k|i] \mat{R} \vec{u}[k|i]\\
        \textrm{s.t.:} \quad
        &\vec{x}[0|i] = \vec{x}[i],\\
        & \vec{x}[k+1|i] = \mat{A} \vec{x}[k|i] + \mat{B} \vec{u}[k|i],\\
        & k \in \mathbb{Z}_{+},
    \end{aligned}
    \label{eq:lqr}
\end{equation}
where $\mat{A} \coloneqq \left. \frac{\partial f}{\partial \vec{x}} \right\rvert_{(\vec{0}, \vec{0})} \in \mathbb{R}^{n \times n}, \mat{B} \coloneqq \left. \frac{\partial f}{\partial \vec{u}} \right\rvert_{(\vec{0}, \vec{0})} \in \mathbb{R}^{n \times m}$ are linearized system dynamics around the equilibrium, and
$\mat{Q} \in \{\mat{N} \in \mathbb{R}^{n \times n} \mid \mat{N} \succeq 0\}, \mat{R} \in \{\mat{M} \in \mathbb{R}^{m \times m} \mid \mat{M} \succeq 0\}$ are weighting matrices that measure significance of state deviations and control costs, respectively.

It can be shown that the optimal LQR control law for an initial state $\vec{x}$ is a linear map 
\begin{equation}
    u_{\textnormal{LQR}}(\vec{x}) = -\mat{K}\vec{x}
\end{equation}
for some $\mat{K} \in \mathbb{R}^{m \times n}$ \cite{anderson2018optimal}. 
Furthermore, if the linearized system $(\mat{A}, \mat{B})$ is \textit{stabilizable} and if $\mat{A} - \mat{B}\mat{K}$ is Hurwitz, then the closed-loop system
\begin{equation}
    \vec{x}[i+1] = f(\vec{x}[i], u_{\textnormal{LQR}}(\vec{x}[i])),
    \label{eq:clsys-lqr}
\end{equation}
is \textit{locally} asymptotically stable near the equilibrium with a \textit{positively invariant} region of attraction $\mathcal{R}_{\textnormal{LQR}}$, as shown in Theorem 4.7 in ~\cite{khalil2002nonlinear}. 
Specifically, $\mathcal{R}_{\textnormal{LQR}}$ is a set that if $\vec{x}_{0} \in \mathcal{R}_{\textnormal{LQR}}$ and $\vec{x}[0] = \vec{x}_{0}$, then $\vec{x}[i] \in \mathcal{R}_{\textnormal{LQR}}, \forall i \in \mathbb{Z}_{>0}$ and $\lim_{i\rightarrow\infty}\vec{x}[i] = \vec{0}$.
As implied by local asymptotic stability, the largest possible $\mathcal{R}_{\textnormal{LQR}}$ is nonempty.

Because the LQR controller only requires one matrix multiplication within a single controller step, it is usually significantly more efficient than the implicit MPC in terms of per-step computation.
Hence, it is usually economical to switch to the LQR controller when the state is in the region of attraction of the LQR, that is, $\vec{x}[i] \in \mathcal{R}_{\textnormal{LQR}}$.

\subsubsection{NN from MPC}
\label{subsubsec:nn-from-mpc}
Provided with an MPC, we derive a NN controller 
\textsl{1)} by imitating the MPC policy $u_{\textnormal{MPC}}(\vec{x})$ through supervised learning and 
\textsl{2)} \textit{optionally} by interacting with the environment through reinforcement learning, as illustrated in Figure~\ref{fig:learning-schematics}.

\begin{figure}
    \begin{tikzpicture}
        \node[align=center] (A) at (0,2) {$u_{\textnormal{MPC}}(\vec{x})$};
        \node[align=center] (B) at (0,0) {$u_{\textnormal{NN}}(\vec{x})$};
        \node[align=center] at (0,-0.75) {(a) Imitation phase};
        \draw[->,>=latex] (A) -- (B) node[midway,sloped,left,rotate=90,text width=2cm,align=right] {};

        \node[align=center] (C) at (4,2) {$f(\vec{x}, u_{\textnormal{NN}}(\vec{x}))$};
        \node[align=center] (D) at (4,0) {$u_{\textnormal{NN}}(\vec{x})$};
        \node[align=center] at (4,-0.75) {(b) Adaptation phase};
        \draw[->,>=latex] (C) -| (2,1) |- (D);
        \draw[->,>=latex] (D) -| (6,1) |- (C);
    \end{tikzpicture}
    \caption{
        Schematics of the learning processes.
        (a) Imitation phase: $(\vec{x}, u_{\textnormal{MPC}}(\vec{x}))$ provides data to train the NN controller $u_{\textnormal{NN}}$ in a supervised learning setting. 
        (b) Adaptation phase (optional): $u_{\textnormal{NN}}$ interacts with the environment $f(\cdot, \cdot)$ to enhance control effectiveness in a reinforcement learning setting.
    }
    \label{fig:learning-schematics}
\end{figure}
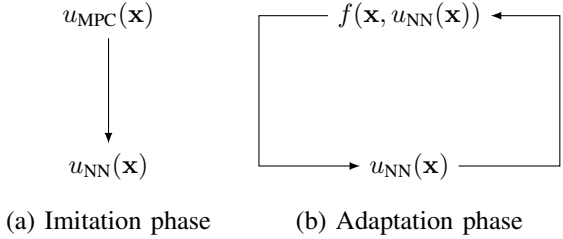

In the \textit{imitation} phase, the NN training is a supervised learning problem usually of the following mean squared form
\begin{equation}
  \min_{\mathcal{W}} \quad \frac{1}{M} \sum_{\vec{x} \in \mathcal{D}} \left\Vert u_{\textnormal{MPC}}(\vec{x}) - \phi(\vec{x} \mid \mathcal{W}) \right\Vert_{2}^{2},
  \label{eq:nn}
\end{equation}
where $\phi : \mathbb{R}^{n} \rightarrow \mathbb{R}^{m}$ is a NN model with weights $\mathcal{W}$ and $\mathcal{D}$ is a set of $M$ states randomly sampled from $\mathbb{X}_{0}$.

Therefore, the resulting NN control law is
\begin{equation}
    u_{\textnormal{NN}}(\vec{x}) \coloneqq \phi(\vec{x} \mid \overline{\mathcal{W}}),
\end{equation}
where $\overline{\mathcal{W}}$ is a suboptimal solution to problem~\eqref{eq:nn}, generally estimated via a gradient method such as Adam~\cite{DBLP:journals/corr/KingmaB14}.

In the \textit{adaptation} phase, it is possible to fine tune the NN further through a reinforcement learning (RL) algorithm.
The training setup varies based on the specific choice of the RL algorithm.
A comprehensive review of reinforcement learning is clearly beyond the scope of this work.
Interested readers may refer to the standard text on this topic~\cite{sutton2018reinforcement}.
Lastly, we remark that this phase may be skipped if the NN trained in the imitation phase is sufficiently effective for the target control application or if the RL algorithm does not bring any improvement to the trained NN.

In practice, people find that implicit MPC controllers can be well approximated by a more efficient NN controller~\cite{parisini1995receding}.
However, it is not easy to certify stability of the closed-loop system without bounding approximation errors between the NN controller and the original MPC controller. 
\begin{equation}
    \vec{x}[i+1] = f(\vec{x}[i], u_{\textnormal{NN}}(\vec{x}[i])).
    \label{eq:clsys-nn}
\end{equation}
To solve the problem, we show that by composing the NN controller with the MPC and the LQR in a hybrid control scheme, we can prove the local asymptotic stability of the closed-loop system, even if the NN is \textit{random}.
Next, we present the basic form of our method, i.e., the standard MAMPC.

\subsubsection{Standard Memory-Augmented MPC}
\label{subsubsec:std-mampc}
The hybrid control scheme combines the LQR controller and the NN controller with the original MPC controller.
At every control step, if the state is in the region of attraction of the LQR, apply the LQR controller; 
else, we simulate the closed-loop system~\eqref{eq:clsys-nn} for up to $N_{\textnormal{LQR}}$ steps: if there \textit{exists} a step $j \leq N_{\textnormal{LQR}}$ such that the state reaches within the region of attraction of LQR and that up until the $j^\text{th}$ step the simulated system does not violate any stage constraint of the MPC, apply the NN controller;
otherwise, if the state is within the admissible initial states of the MPC, apply the implicit MPC.

Formally, the hybrid controller is defined as follows
\begin{equation}
    \begin{aligned}
        &u_{\textnormal{MAMPC}}(\vec{x}) \coloneqq \nonumber \\
        &\begin{cases}
            u_{\textnormal{LQR}}(\vec{x}),  & \textnormal{if } \vec{x} \in \mathcal{R}_{\textnormal{LQR}},\\
            \hline
                                            & \textnormal{if } \vec{x} \in \mathbb{X}_{0} \setminus \mathcal{R}_{\textnormal{LQR}} \quad \textnormal{and}\\
            u_{\textnormal{NN}}(\vec{x}),   & \quad \exists i = 1, \dots, N_{\textnormal{LQR}}, \quad \vec{y}[i]\in \mathcal{R}_{\textnormal{LQR}} \quad \textnormal{and}\\
                                            & \quad \forall j = 0, \dots, i, \quad (\vec{y}[j], u_{\textnormal{NN}}(\vec{y}[j])) \in \mathbb{A},\\
            \hline
            u_{\textnormal{MPC}}(\vec{x}),  & \text{otherwise},
        \end{cases}
        \label{eq:mampc}
    \end{aligned}    
\end{equation}

where $\vec{y}$ is the \textit{simulated} state by numerically stepping: $\vec{y}[i] = f(\vec{y}[i-1], u_{\textnormal{NN}}(\vec{y}[i-1]))$ with $\vec{y}[0] = \vec{x}$, 
$\mathcal{R}_{\textnormal{LQR}}$ is designed to be a subset of $\mathbb{X}_{0}$, and 
$N_{\textnormal{LQR}} \in \mathbb{Z}_{> 0}$ is the verification horizon of $\mathcal{R}_{\textnormal{LQR}}$.

As we show in Section~\ref{sec:theoretical-analysis}, the stability of NN relies on verifying whether the system reaches within $\mathcal{R}_{\textnormal{LQR}}$ through forward simulation.
However, this approach is ineffective for systems that are sensitive to initial conditions and systems that require a large number of steps to converge to the origin.
To address these two challenges, we introduce the following two variants of MAMPC.

\subsubsection{Alternating-Authority Memory-Augmented MPC}
\label{subsubsec:aa-mampc}
The first variant of MAMPC is designed specifically for chaotic systems.
A \textit{chaotic system} is loosely defined as a system that is very sensitive to initial condition and control input.
Qualitatively speaking, the same chaotic system that begins with two slightly different initial conditions and control sequences will arrive at significantly different terminal states. 
As a result, it is very challenging to stabilize a chaotic system with any NN controller because the function approximator will inevitably produce random control errors which can easily deter the system from its stabilizing trajectory.
In this case, it is useful to modify the hybrid control scheme to bound error accumulation induced by the NN by periodically alternating between the NN controller and the MPC controller.
We name this modified version of MAMPC as \textit{alternating-authority} MAMPC, which is formally described below.
\begin{equation}
    \begin{aligned}
        &u_{\textnormal{MAMPC}}^{\textnormal{AA}}(\vec{x}) \coloneqq \nonumber \\
        &\begin{cases}
            u_{\textnormal{LQR}}(\vec{x}),  & \textnormal{if }\vec{x} \in \mathcal{R}_{\textnormal{LQR}},\\
            \hline
                                            & \textnormal{if }\vec{x} \in \mathbb{X}_{0} \setminus \mathcal{R}_{\textnormal{LQR}} \quad \textnormal{and}\\
            u_{\textnormal{NN}}(\vec{x}),   & \quad u_{\textnormal{NN}}(\vec{x}) \in \mathbb{U} \land \vec{y}[1] \in \mathbb{X} \quad \textnormal{and}\\
                                            & \quad i \Mod{i_{d}} \neq 0,\\
            \hline
            u_{\textnormal{MPC}}(\vec{x}), & \text{otherwise},
        \end{cases}
        \label{eq:mampc-variant1}
    \end{aligned}
\end{equation}
where $i_{d} \in \mathbb{Z}_{\geq 1}$ is the period of MPC defaulting and $\mathcal{R}_{\textnormal{LQR}} \subset \mathbb{X}_{0}$.

The main variation from the standard MAMPC is that we do not verify if the system falls within $\mathcal{R}_{\textnormal{LQR}}$ in $N_{\textnormal{LQR}}$ steps, as forward simulation is not as reliable for chaotic systems.
This implies that $i_{d}$ should be chosen as a small positive number.
Otherwise, prediction produced by forward simulation will be a reliable estimate of the future.
However, note that reducing $i_{d}$ will prolong running time.
Therefore, a rule of thumb of designing $i_{d}$ should be to choose it as large as the stability permits.

In addition, we remark that the closed-loop system has a smaller region of attaction compared to standard MAMPC, since we do not verify whether the system enters $\mathcal{R}_{\textnormal{LQR}}$ through forward simulation.

\subsubsection{Way-Point Memory-Augmented MPC}
\label{subsec:wp-mampc}
Another variant of MAMPC is designed specifically for slow systems.
A \textit{slow system} is loosely defined as a system that requires a substantial number of steps to steer close to the origin.
For slow systems, it is difficult to choose an effective verification horizon $N_{\textnormal{LQR}}$: a small $N_{\textnormal{LQR}}$ may be too short to predict convergence of the NN controller; a large $N_{\textnormal{LQR}}$ may incur computational overhead and thus render the hybrid control inefficient.
To address this problem, we propose a second variant of MAMPC by introducing a ``way-point'' set, $\mathcal{D}_{\textnormal{WP}}$, and enabling the NN controller to be applied as long as the anticipated trajectory falls with that way-point set.
We name this variant of MAMPC as \textit{way-point} MAMPC, which is formulated as follows.
\begin{equation}
    \begin{aligned}
        &u_{\textnormal{MAMPC}}^{\textnormal{WP}}(\vec{x}) \coloneqq \nonumber \\
        &\begin{cases}
            u_{\textnormal{LQR}}(\vec{x}), & \textnormal{if }\vec{x} \in \mathcal{R}_{\textnormal{LQR}},\\
            \hline
                                          & \textnormal{if }\vec{x} \in \mathcal{D}_{\textnormal{WP}} \setminus \mathcal{R}_{\textnormal{LQR}} \quad \textnormal{and}\\
            u_{\textnormal{NN}}(\vec{x}),  & \quad \exists i = 1, \dots N_{\textnormal{LQR}}, \quad \vec{y}[i]\in \mathcal{R}_{\textnormal{LQR}} \quad \textnormal{and}\\
                                          & \quad \forall j = 0, \dots i, \quad (\vec{y}[j], u_{\textnormal{NN}}(\vec{y}[j])) \in \mathbb{A}_{\textnormal{WP}},\\
            \hline
                                          & \textnormal{if }\vec{x} \in \mathbb{X}_{0} \setminus \mathcal{D}_{\textnormal{WP}} \quad \textnormal{and}\\
            u_{\textnormal{NN}}(\vec{x}),  & \quad \exists i \in 1, \dots, N_{\textnormal{WP}}, \quad \vec{y}[i]\in \mathcal{D}_{\textnormal{WP}} \quad \textnormal{and}\\
                                          & \quad \forall j \in 0, \dots, i, \quad (\vec{y}[j], u_{\textnormal{NN}}(\vec{y}[j])) \in \mathbb{A},\\
            \hline
            u_{\textnormal{MPC}}(\vec{x}), & \text{otherwise},
        \end{cases}
        \label{eq:mampc-variant2}
    \end{aligned}
\end{equation}
where $\mathcal{D}_{\textnormal{WP}}$ is a compact set that contains $\mathcal{R}_{\textnormal{LQR}}$, $\mathbb{A}_{\textnormal{WP}} \coloneqq \mathcal{D}_{\textnormal{WP}} \times \mathbb{U}$, and $N_{\textnormal{WP}} \in \mathbb{Z}_{>0}$ is the verification horizon of $\mathcal{D}_{\textnormal{WP}}$.

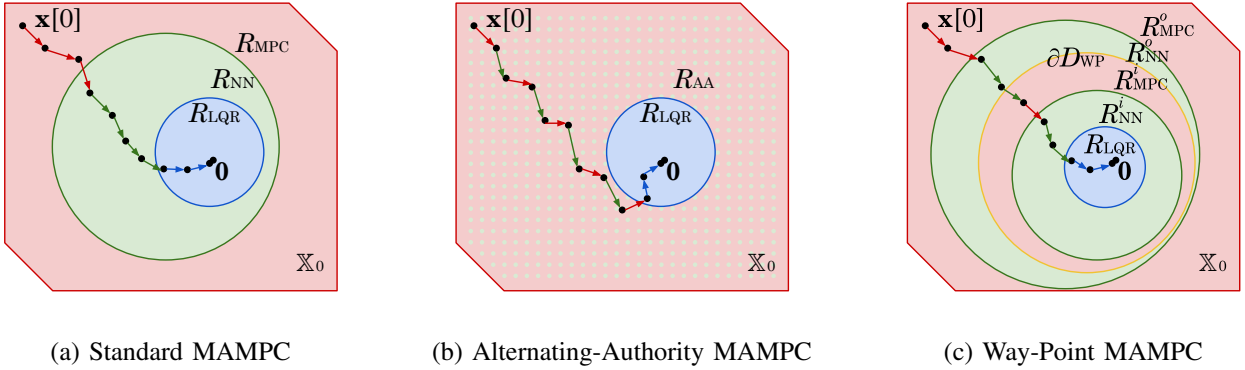
\begin{figure*}
    \centering
    \begin{tikzpicture}
        \node[inner sep=0pt] (standard) at (0,0)
            {\includegraphics[width=.3\textwidth]{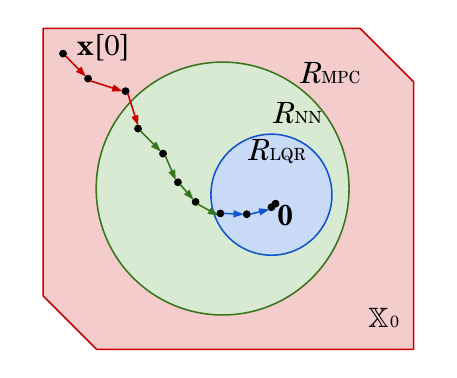}};
        \node[align=center] at (0,-2.75) {(a) Standard MAMPC};
        \node[inner sep=0pt] (variant1) at (6,0)
            {\includegraphics[width=.3\textwidth]{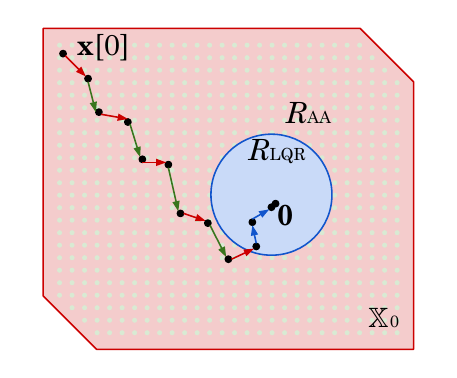}};
        \node[align=center] at (6,-2.75) {(b) Alternating-Authority MAMPC};
        \node[inner sep=0pt] (variant2) at (12,0)
            {\includegraphics[width=.3\textwidth]{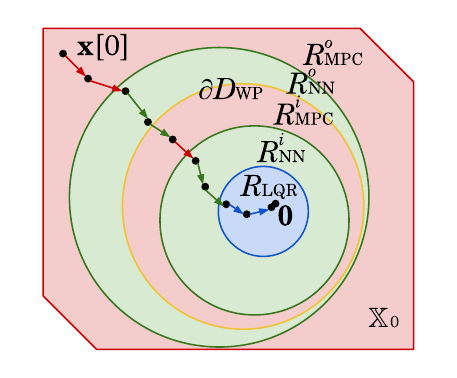}};
        \node[align=center] at (12,-2.75) {(c) Way-Point MAMPC};
    \end{tikzpicture}
    \caption{        
        An illustrative example of how different variants of MAMPC \textit{could} steer a system toward equilibrium.
        (a) Standard MAMPC: a triple-mode hybrid control composed of an MPC, a NN, and a LQR. 
        (b) Alternating-authority MAMPC: a variant of the standard MAMPC, modified for chaotic systems. 
        (c) Way-point MAMPC: a variant of the standard MAMPC, modified for slow systems. 
        Solid red regions represent $\mathbb{X}_{0}$; solid green regions represent $\mathcal{R}_{\textnormal{NN}}$; \textit{green dots} indicate periodic controller switching region in alternating-authority MAMPC; and solid blue regions represent $\mathcal{R}_{\textnormal{LQR}}$.
        A black dot represents a state and an arrow represents a controlled state transition.
        The color of arrows indicates which sub-controller is invoked during that control step: red is MPC; green is NN; and blue is LQR.
        The yellow line represents $\partial{\mathcal{D}_{\textnormal{WP}}}$.
        Note that the above examples by no means exhaustive.
        For instance, it is possible that $\mathcal{R}_{\textnormal{NN}} = \emptyset$.
    }
    \label{fig:mampc-convergence}
\end{figure*}

At minimum, we require that the way-point set contains the region of attraction of the LQR, that is, $\mathcal{R}_{\textnormal{LQR}} \subset \mathcal{D}_{\textnormal{WP}}$.
Additional improvement of the way-point set design can be performed through search.
For example, one can parameterize $\mathcal{D}_{\textnormal{WP}}$ with one or more parameters and search for a set of parameters that produces the most efficient closed-loop performance.

Besides, note that introduction of a way-point set may reduce the area of the region of attraction of the closed-loop system.


\section{Theoretical Analysis}
\label{sec:theoretical-analysis}

In this section, we prove the stability of the three MAMPC methods proposed above.
We also remark on the robustness of the methods and the necessity of the fail-safe MPC mode.

\subsection{Stability of MAMPC}
\label{subsec:proof-for-stability}
We show that the three MAMPC methods are locally asymptotically stable with slightly different stability properties.
We first prove that standard MAMPC is locally asymptotically stable in $\mathbb{X}_{0}$, as stated in the theorem below.
\begin{theorem}[Stability of Standard MAMPC]
    For every MPC problem of the form~\eqref{eq:mpc}, there always exists a $u_{\textnormal{MAMPC}}$ of the form~\eqref{eq:mampc}.
    Furthermore, the closed-loop system
    \begin{equation}
        \vec{x}[i+1] = f\left(\vec{x}[i], u_{\textnormal{MAMPC}}(\vec{x}[i])\right), \quad i \in \mathbb{Z}_{\geq 0},
        \label{eq:clsys-mampc}
    \end{equation}
    with $\vec{x}[0] = \vec{x}_{0} \in \mathbb{X}_{0}$ is locally asymptotically stable in $\mathbb{X}_{0}$.
    \label{thm:mampc-stability}
\end{theorem}
\begin{proof}
    See Appendix~\ref{app:mampc-proof}.
\end{proof}

To provide geometric intuition on how a MAMPC achieves local asymptotic stability, we illustrate schematically in Figure~\ref{fig:mampc-convergence}a how a MAMPC \textit{could} steer a system from an initial state to the equilibrium at zero.
A state in $\mathcal{R}_{\textnormal{MPC}}$ is first steered by $u_{\textnormal{MPC}}$ into $\mathcal{R}_{\textnormal{NN}}$, which is then steered by $u_{\textnormal{NN}}$ into $\mathcal{R}_{\textnormal{LQR}}$, after which is steered to zero by $u_{\textnormal{LQR}}$.
Note that the example demonstrated in Figure~\ref{fig:mampc-convergence}a is only a hypothetical particular case of MAMPC closed-loop behaviors.

Next, for alternating-authority MAMPC, the closed-loop system no longer has guaranteed local asymptotic stability in $\mathbb{X}_{0}$, because it does not check whether NN shoots inside $\mathcal{R}_{\textnormal{LQR}}$.
Rather, we can only show that the system is locally asymptotically stable in $\mathcal{R}_{\textnormal{LQR}}$ and will never escape $\mathbb{X}_{0}$, as stated the theorem below.
\begin{theorem}[Stability of Alternating-Authority MAMPC]
    For every MPC problem of the form~\eqref{eq:mpc}, there always exists a $u_{\textnormal{MAMPC}}^{\textnormal{AA}}$ of the form~\eqref{eq:mampc-variant1}.
    Furthermore, the closed-loop system 
    \begin{equation}
        \vec{x}[i+1] = f\left(\vec{x}[i], u_{\textnormal{MAMPC}}^{\textnormal{AA}}(\vec{x}[i])\right), \quad i \in \mathbb{Z}_{\geq 0},
        \label{eq:clsys-aa-mampc}
    \end{equation}
    with $\vec{x}[0] = \vec{x}_{0} \in \mathbb{X}_{0}$ is locally asymptotically stable in $\mathcal{R}_{\textnormal{LQR}}$. 
    Besides, for every initial condition $\vec{x}[0] \in \mathbb{X}_{0}$, $\vec{x}[i] \in \mathbb{X}_{0}$ for all $i \in \mathbb{Z}_{+}$.
    \label{thm:aa-mampc-stability}
\end{theorem}
\begin{proof}
    See Appendix~\ref{app:aa-mampc-proof}.
\end{proof}

An illustration of how the alternating-authority MAMPC \textit{could} steer a system to the equilibrium is provided in Figure~\ref{fig:mampc-convergence}b.
Note that marginal stability may be an understatement in practice:
a well-trained NN, combined with a modest alternation period $i_{d}$, could very well result in a closed-loop system that is asymptotically stable.

If analytical guarantee of asymptotic stability in the \textit{entire} $\mathbb{X}_{0}$ is absolutely required, one can replace the admissible control set $\mathbb{X}_{0}$ with a time-varying set $\mathbb{S}[i]$, where 
\begin{equation}
    \mathcal{R}_{\textnormal{LQR}} \subset \mathbb{S}[i+1] \subset \mathbb{S}[i] \subset \mathbb{X}_{0}, \quad \forall i = 1, \dots, n-1,
    \label{eq:shrink}
\end{equation}
for some $n \in \mathbb{Z}_{>1}$ with $\mathbb{S}[0] = \mathbb{X}_{0}$ and $\mathbb{S}[n] = \mathcal{R}_{\textnormal{LQR}}$.
Geometrically, $\mathbb{S}$ can be viewed as a set that gets successively shrunk from $\mathbb{X}_{0}$ toward $\mathcal{R}_{\textnormal{LQR}}$.
Such design will invoke the fail-safe MPC if the systems is ``stuck'' outside of $\mathcal{R}_{\textnormal{LQR}}$.

Lastly, similar to alternating-authority MAMPC, the closed-loop system of a way-point MAMPC is locally asymptotically stable in $\mathcal{R}_{\textnormal{LQR}}$ and will eventually enter and stay within a level set that \textit{escribes} $\mathcal{D}_{\textnormal{WP}}$, as stated in the following theorem.
\begin{theorem}[Stability of Way-Point MAMPC]
    For every MPC problem of the form~\eqref{eq:mpc}, there always exists a $u_{\textnormal{MAMPC}}^{\textnormal{WP}}$ of the form~\eqref{eq:mampc-variant2}.
    Furthermore, the closed-loop system 
    \begin{equation}
        \vec{x}[i+1] = f\left(\vec{x}[i], u_{\textnormal{MAMPC}}^{\textnormal{WP}}(\vec{x}[i])\right), \quad i \in \mathbb{Z}_{\geq 0},
        \label{eq:clsys-wp-mampc}
    \end{equation}
    with $\vec{x}[0] = \vec{x}_{0} \in \mathbb{X}_{0}$ is locally asymptotically stable in $\mathcal{R}_{\textnormal{LQR}}$.
    Besides, there exists a time index $j \in \mathbb{Z}_{\geq 0}$ such that for all $i \geq j$
    \begin{equation*}
        \vec{x}[i] \in \mathcal{D}_{\textnormal{WP}}^{+} \coloneqq \{\vec{x} \in \mathbb{R}^{n} \mid V(\vec{x}) < \max_{q \in \partial{\mathcal{D}_{\textnormal{WP}}}} V(\vec{x})\},
    \end{equation*}
    where $V: \mathbb{R}^{n} \rightarrow \mathbb{R}_{+}$ is a Lyapunov function of the closed-loop system~\eqref{eq:clsys-mpc} and $\partial{\mathcal{D}_{\textnormal{WP}}}$ is the boundary of the compact set $\mathcal{D}_{\textnormal{WP}}$ that by construction contains $\mathcal{R}_{\textnormal{LQR}}$.
    \label{thm:wp-mampc-stability}
\end{theorem}
\begin{proof}
    See Appendix~\ref{app:wp-mampc-proof}.
\end{proof}

An illustration of how the way-point MAMPC \textit{could} steer a system to the equilibrium is provided in Figure~\ref{fig:mampc-convergence}c.
As in the alternating-authority MAMPC, to always achieve asymptotic stability, one can successively shrink $\mathcal{D}_{\textnormal{WP}}$ until $\mathcal{D}_{\textnormal{WP}} = \mathcal{R}_{\textnormal{LQR}}$ in the same way as presented in equation \eqref{eq:shrink}.

\subsection{Remark on Robustness}
\label{subsec:remark-on-robustness}

MAMPC may be \textit{robustified} to account for model uncertainties using the idea of bounding error margin.
We demonstrate such robustification procedure by applying it to the standard MAMPC.

Let the \textit{true} system dynamics be $f$.
We distinguish $\bar{f}$ as the \textit{model} of the dynamical system.
Define the model uncertainties of a controller $u$ with initial condition $\vec{x}[i]$ at time step $i+k$ as
\begin{equation*}
    \Delta \vec{x}_{u}[k|i] \coloneqq \vec{y}[k] - \bar{\vec{y}}[k],
\end{equation*}
where
\begin{equation*}
    \begin{aligned}
        \vec{y}[0] &= \vec{x}[i], \quad \vec{y}[k] = f(\vec{y}[k-1], u(\vec{y}[k-1])),\\
        \bar{\vec{y}}[0] &= \vec{x}[i], \quad \bar{\vec{y}}[k] = \bar{f}(\bar{\vec{y}}[k-1], u(\bar{\vec{y}}[k-1])).
    \end{aligned}
\end{equation*}

When $\Vert \Delta \vec{x}_{u}[k|i] \Vert = \alpha > 0$, it is possible that $\vec{y}[k] \notin \mathcal{R}_{\textnormal{LQR}}$ but $\bar{\vec{y}}[k] \in \mathcal{R}_{\textnormal{LQR}}$, potentially leading to a failure of MAMPC.
To address this problem, we define a set operator $\Ero_{\delta}$ as follows
\begin{equation*}
    \Ero_{\delta}(\mathcal{X}) \coloneqq \{ \vec{x} \in \mathcal{X} \mid \Vert \vec{x} - \vec{q} \Vert \geq \delta, \forall \vec{q} \in \partial\mathcal{X}\} \textnormal{ for some } \delta > 0,
\end{equation*}
where $\mathcal{X}$ is an arbitrary set and $\partial\mathcal{X}$ denotes the boundary of the set $\mathcal{X}$.
We claim that for $\Vert \Delta \vec{x}_{u}[k|i] \Vert = \alpha > 0$, if $\bar{\vec{y}}[k] \in \Ero_{\delta}(\mathcal{R}_{\textnormal{LQR}})$ with $\delta \geq \alpha$, then there must be $\vec{y}[k] \in \mathcal{R}_{\textnormal{LQR}}$.
Proof of the claim is a direct application of the triangle inequality of norm.
We can therefore incorporate the set operator $\Ero_{\delta}$ to enhance the robustness of MAMPC with respect to model uncertainties.

We can apply the above procedure to robustify the standard MAMPC as follows
\begin{equation}
    \begin{aligned}
        &u_{\textnormal{MAMPC}}^{\delta}(\vec{x}) \coloneqq \nonumber \\
        &\begin{cases}
            u_{\textnormal{LQR}}(\vec{x}), & \vec{x} \in \mathcal{R}_{\textnormal{LQR}},\\
            \hline
                                           & \vec{x} \in \mathbb{X}_{0} \setminus \mathcal{R}_{\textnormal{LQR}},\\
            u_{\textnormal{NN}}(\vec{x}),  & \exists i = 1, \dots, N_{\textnormal{LQR}}, \quad \bar{\vec{y}}[i]\in \mathcal{R}_{\textnormal{LQR}}^{\delta},\\
                                           & \forall j = 0, \dots, j, \quad (\bar{\vec{y}}[j], u_{\textnormal{NN}}(\bar{\vec{y}}[j])) \in \mathbb{A}^{\delta}\\
            \hline
            u_{\textnormal{MPC}}(\vec{x}), & \text{otherwise},
        \end{cases}
    \end{aligned}
\end{equation}
where $\mathcal{R}_{\textnormal{LQR}}^{\delta} \coloneqq \Ero_{\delta}(\mathcal{R}_{\textnormal{LQR}})$ and $\mathbb{A}^{\delta} \coloneqq \Ero_{\delta}(\mathbb{X}) \times \mathbb{U}$, for some $\delta > 0$.
The stability of the above robustified MAMPC is stated in the following theorem.
\begin{theorem}[Robustification of Standard MAMPC]
    Suppose the model uncertainties are upper bounded by some $\alpha > 0$, that is,
    \begin{equation*}
        \Vert \Delta \vec{x}_{u}[k|i] \Vert \leq \alpha, \quad \forall i \in \mathbb{Z}_{\geq 0}, \forall k = 0, \dots, N_{\textnormal{LQR}}.
    \end{equation*}
    Then the following closed-loop system
    \begin{equation*}
        \vec{x}[i+1] = f(\vec{x}[i], u_{\textnormal{MAMPC}}^{\alpha}(\vec{x}[i])), \quad i \in \mathbb{Z}_{\geq 0},
    \end{equation*}
    with $\vec{x}[0] = \vec{x}_{0} \in \mathbb{X}_{0}$ is locally asymptotically stable in $\mathbb{X}_{0}$.
    \label{thm:robust-mampc}
\end{theorem}

\begin{proof}
    The proof is identical to that of Theorem~\ref{thm:mampc-stability}.
    Note that it possible that $\mathcal{R}_{\textnormal{LQR}}^{\delta} = \emptyset$.
    In this case, the hybrid control is virtually equivalent to a dual-mode MPC-LQR controller.
\end{proof}
Similar to the above example, the robustification procedure can also be applied to enhance the alternating-authority MAMPC and way-point MAMPC.

\subsection{Remark on Necessity of Fail-Safe MPC}
\label{subsec:remark-on-necessity-of-fail-safe-mpc}
The hybrid control scheme can be clearly more efficient if we can skip the forward verification and remove the fail-safe MPC mode.
In light of this observation, we describe a condition, where the fail-safe MPC mode will \textit{never} be invoked and thus can be safely removed from the hybrid control scheme.

Define the following ``fail-free'' MAMPC
\begin{equation*}
    \begin{aligned}
        u_{\textnormal{MAMPC}}^{+}(\vec{x}) \coloneqq 
        \begin{cases}
            u_{\textnormal{LQR}}(\vec{x}),       & \vec{x} \in \mathcal{R}_{\textnormal{LQR}},\\
            u_{\textnormal{NN}}(\vec{x})         & \vec{x} \in \mathbb{X}_{0} \setminus \mathcal{R}_{\textnormal{LQR}}.
        \end{cases}
    \end{aligned}
\end{equation*}
Let $L: \mathbb{X} \rightarrow \mathbb{R}_{+}$ be an augmented cumulative cost function of NN, that is,
\begin{equation*}
    L(\vec{x}) \coloneqq \sum_{k=0}^{N-1}{c(\vec{y}[k], u_{\textnormal{NN}}(\vec{y}[k]))} + c_{f}(\vec{y}[N]),
\end{equation*}
with
\begin{equation*}
    \vec{y}[0] = \vec{x}, \quad \vec{y}[k] = f(\vec{y}[k-1], u_{\textnormal{NN}}(\vec{y}[k-1])).
\end{equation*}

From Theorem 13.1 in~\cite{borrelli2017predictive}, we are ready to state the condition in the following theorem.
\begin{theorem}[Fail-Free MAMPC]
    If there exists some nonnegative function $\gamma: \mathbb{X}_{0} \rightarrow \mathbb{R}_{+}$ such that
    \begin{equation}
        \begin{aligned}
            (\vec{x}, u_{\textnormal{NN}}(\vec{x})) \in \mathbb{A}, \quad &\forall \vec{x} \in \mathbb{X}_{0}, \\
            L(\vec{x}) \leq J^{*}(\vec{x}) + \gamma(\vec{x}), \quad &\forall \vec{x} \in \mathbb{X}_{0}, \\
            \gamma(\vec{x}) - c(\vec{x}, \vec{0}) < 0, \quad &\forall \vec{x} \in \mathbb{X}_{0} \setminus \vec{0}, 
        \end{aligned}
        \label{eq:nn-bound}
    \end{equation}
    then the closed-loop system 
    \begin{equation*}
        \vec{x}[i+1] = f(\vec{x}[i], u_{\textnormal{MAMPC}}^{+}(\vec{x}[i])),  \quad i \in \mathbb{Z}_{\geq 0},
    \end{equation*}
    with $\vec{x}[0] = \vec{x}_{0} \in \mathbb{X}_{0}$ is locally asymptotically stable in $\mathbb{X}_{0}$ without violating any constraints in the original MPC problem~\eqref{eq:mpc}.
    \label{thm:2mampc-nofs}
\end{theorem}
\begin{proof}
The proof of the theorem is a direct application of Theorem 13.1 in~\cite{borrelli2017predictive}.
\end{proof}
Note that the inequality conditions in~\eqref{eq:nn-bound} are typically verified through sampling and interpolation \cite{borrelli2017predictive}, which makes this approach only suitable for relatively simple, moderately sized problems.
Besides, the above theorem is only a sufficient condition and other sufficient conditions are possible too.
For example, see~\cite{richards2018lyapunov,chen2021learning}.

Last but not least, we emphasize that \textit{all} the analysis above does not impose any condition on the approximation error of the NN.
So long as the NN has the right input and output dimensions, the corresponding MAMPC will be closed-loop stable, even if the weights of the NN are \textit{random}.
In the worst case, MAMPC reduces to a hybrid controller consisting of just the MPC and the LQR, where the NN mode will never be invoked.
This reduced hybrid controller is virtually identical to the one proposed in~\cite{michalska1993robust}, despite of having slightly different running time performance.


\section{Numerical Experiments}
\label{sec:numerical-experiments}

To evaluate the running time performance of MAMPC, we conduct four numerical experiments on controlling a pendulum, a triple pendulum, a bicopter, and a quadcopter.
The models are selected for their relevance in industrial robotic applications:
Pendulum models are the building blocks for robot arm manipulation, while copter models are important model classes in unmanned aerial vehicle control.
Performance and efficiency of MAMPC are evaluated through comparison with the corresponding baseline implicit MPC. 

All computations were conducted on an Intel Core i7-1065G7 CPU machine.
Implementation is done via MATLAB.
The following paragraphs highlight the key parameters and results of the four experiments.
For implementation details, please refer to the source code at~\cite{wu2021mampcgithub}.

\subsection{Pendulum}
The first model is a single-arm pendulum as shown in Figure~\ref{fig:pendulum}.
The goal is to maintain the pendulum at the position of the highest potential energy as marked by the dashed line.

The state is defined as $\vec{x} \coloneqq [\theta \quad \dot{\theta}]^{\top} \in [-\pi, \pi) \times \mathbb{R}$, where $\theta$ is the angular displacement from the inverted position.
The input is a scalar $u \in \mathbb{R}$, which is a bounded torque applied at the joint.
The control objective is to steer the pendulum to the origin through the limited torque actuator $u \in [-0.05, 0.05] \hspace{0.5em} \SI{}{\newton\meter}$ at the joint.

\begin{figure}[H]
    \centering
    \includegraphics[width=0.25\textwidth]{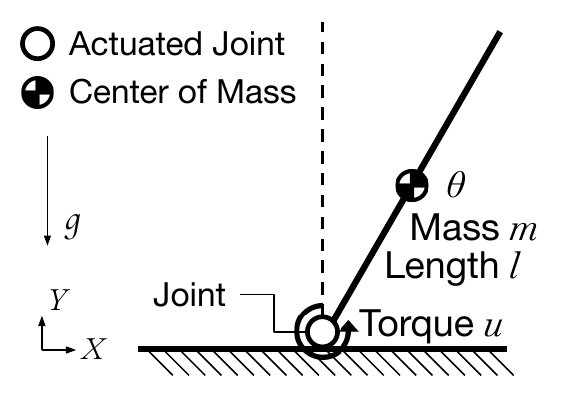}
    \caption{Diagram of a pendulum model. 
    A torque between $[-0.05, 0.05] \hspace{0.5em} \SI{}{\newton\meter}$ is being applied to the hinge to keep the pendulum at its inverted position.}
    \label{fig:pendulum}
\end{figure}

The MPC design highlights a sampling interval of $\SI{0.1}{\second}$, a planning horizon $N = 5$, and three optimization variables per step, resulting in a linearly-constrained quadratic programming (LCQP) of $\SI{0.5}{\second}$ prediction horizon and 15 variables.
Plant nonlinearity is handled through linearizing around the equilibrium, leading to a linear time-invariant system.

We apply a standard MAMPC to control the pendulum with
\begin{equation*}
    \mathcal{R}_{\textnormal{LQR}} = \{\vec{x} \in \mathbb{R}^{2} \mid \Vert \vec{x} \Vert_{2} \leq 0.5\}, \quad N_{\textnormal{LQR}} = 5.
\end{equation*}
The NN is a two-layer, 20-neuron multilayer perceptron (MLP) trained through supervised learning with data uniformly randomly sampled from $\theta \in [\pi, \pi], \dot{\theta} \in [-1, 1]$.

Performance and efficiency of MAMPC, at various stages of training, is illustrated in Figure~\ref{fig:pendulum_exp}.
The total and per-step running times of MAMPC are summarized in the first columns of Table~\ref{tab:sum-total-running-time} and Table~\ref{tab:sum-perstep-running-time}, respectively.

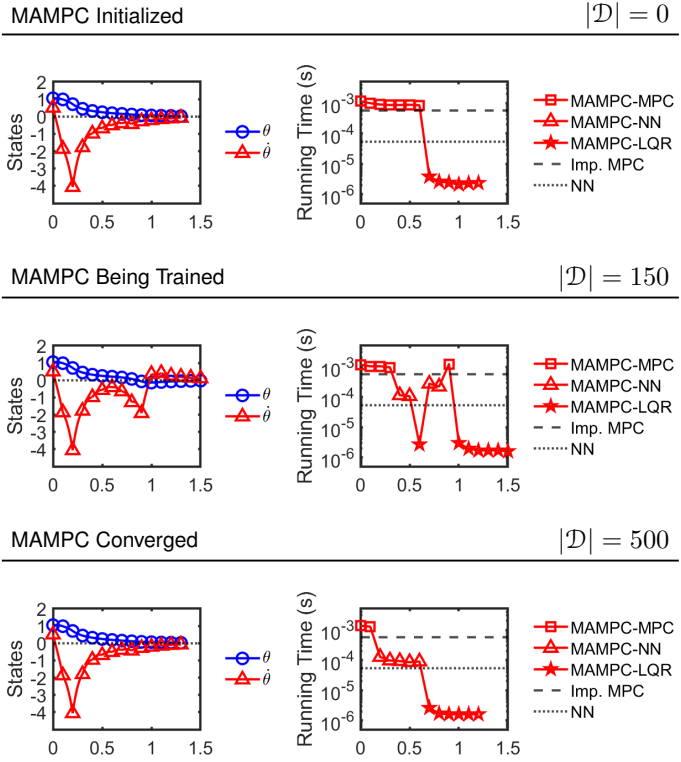
\begin{figure}[H]
    \centering
    \begin{tikzpicture}[font=\sffamily]
        \node[right] at (-4.5, 1.75) {\footnotesize MAMPC Initialized};
        \node[left] at (4.5, 1.75) {$\vert \mathcal{D} \vert = 0$};
        \draw[line width=0.25mm] (-4.5, 1.5) -- (4.5, 1.5);
        \node[inner sep=0pt] (time) at (0,0) {\includegraphics[width=.49\textwidth]{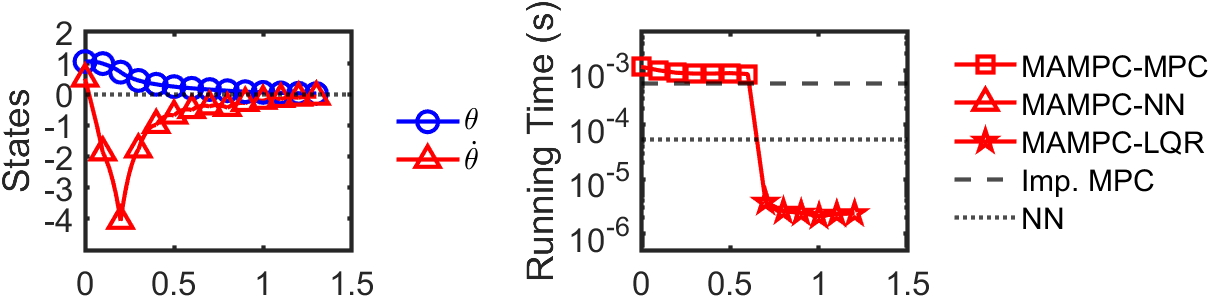}};

        \node[right] at (-4.5, -1.75) {\footnotesize MAMPC Being Trained};
        \node[left] at (4.5, -1.75) {$\vert \mathcal{D} \vert = 150$};
        \draw[line width=0.25mm] (-4.5, -2) -- (4.5, -2);
        \node[inner sep=0pt] (time) at (0,-3.5) {\includegraphics[width=.49\textwidth]{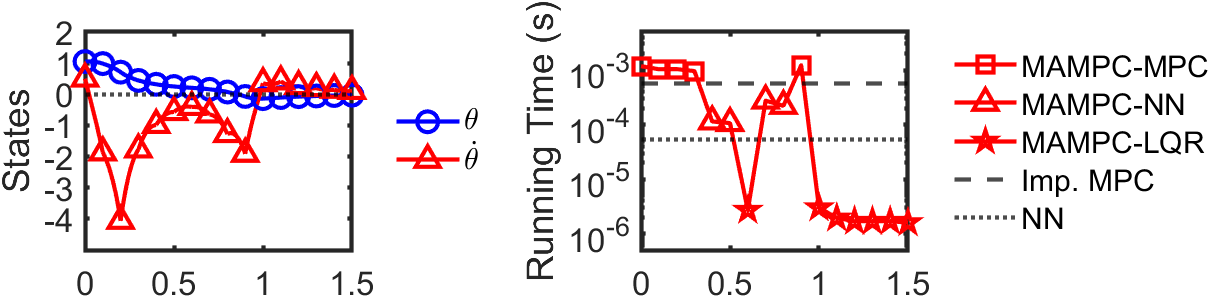}};

        \node[right] at (-4.5, -5.25) {\footnotesize MAMPC Converged};
        \node[left] at (4.5, -5.25) {$\vert \mathcal{D} \vert = 500$};
        \draw[line width=0.25mm] (-4.5, -5.5) -- (4.5, -5.5);
        \node[inner sep=0pt] (time) at (0,-7) {\includegraphics[width=.49\textwidth]{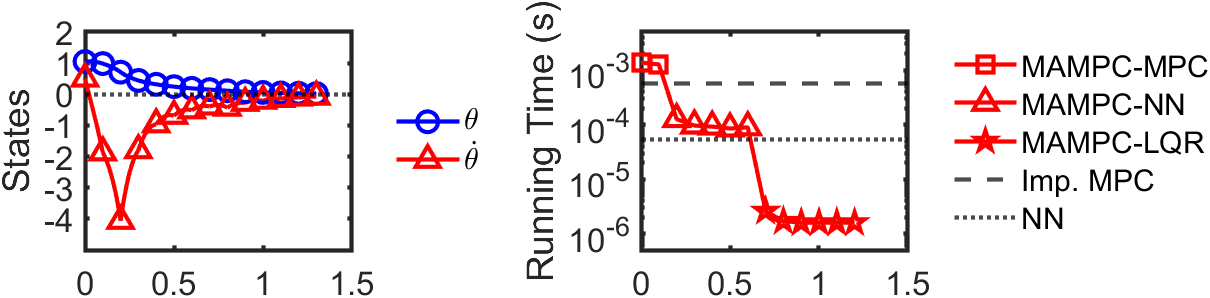}};
    \end{tikzpicture}
    \caption{Application of standard MAMPC to the pendulum problem.
    At initialization, the NN mode fails forward verification and thus never get invoked.
    As the NN is being trained, its likelihood of passing forward verification increases.
    At last, a well trained NN takes over most of MPC steps, resulting in a more efficient operation.}
    \label{fig:pendulum_exp}
\end{figure}

\subsection{Triple Pendulum}

The second model is a triple pendulum as shown in Figure~\ref{fig:triplependulum}, which extends the above pendulum model to a more complex use case.
The goal is to maintain the triple pendulum at the position of the highest potential energy.

The state is defined as $\vec{x} \coloneqq [\theta_{1} \quad \dot{\theta}_{1} \quad \theta_{2} \quad \dot{\theta}_{2} \quad \theta_{3} \quad \dot{\theta}_{3}]^{\top} \in ([-\pi, \pi) \times \mathbb{R})^{3}$, where $\theta_{1}, \theta_{2}, \theta_{3}$ are the angular displacements of the three links respectively.
The input is defined as $\vec{u} \coloneqq [u_{1} \quad u_{2} \quad u_{3}] \in \mathbb{R}^{3}$, which is the applied torques at the three joints respectively.
The control objective is to steer the triple pendulum to the origin by applying three limited torque inputs $u_{1}, u_{2}, u_{3} \in [-1, 1] \hspace{0.5em} \SI{}{\newton\meter}$ at the three joints, respectively.

\begin{figure}[H]
    \centering
    \includegraphics[width=0.30\textwidth]{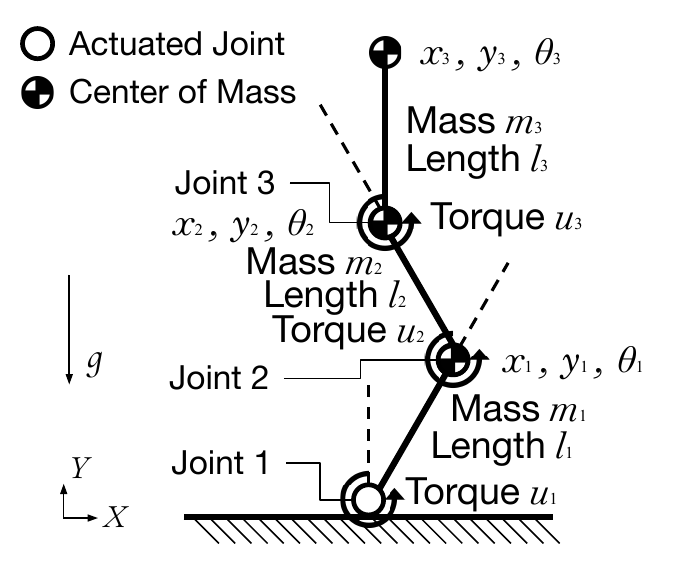}
    \caption{Diagram of a triple pendulum model.
    Note that the centers of mass of the first and second torques coincide with the second and third joints.
    A torque between $[-1, 1] \hspace{0.5em} \SI{}{\newton\meter}$ is being applied to each joint to keep the triple pendulum at its inverted position.}
    \label{fig:triplependulum}
\end{figure}

The MPC design highlights a sampling interval of $\SI{0.1}{\second}$, a planning horizon $N = 5$, and nine optimization variables per step, resulting in a LCQP of $\SI{0.5}{\second}$ prediction horizon and 45 variables.
Linearization is applied as before to handle nonlinearity.

Because triple pendulum is a chaotic system, we apply an alternating-authority MAMPC instead of standard MAMPC with
\begin{equation*}
    \mathcal{R}_{\textnormal{LQR}} = \{\vec{x} \in \mathbb{R}^{6} \mid \Vert \vec{x} \Vert_{2} \leq 0.4\}, \quad N_{\textnormal{LQR}} = 5, \quad i_{d} = 2.
\end{equation*}
The NN features a three-layer, 50-neuron MLP trained through supervised learning with data uniformly randomly sampled from 
\begin{equation*}
    \theta_{1}, \theta_{2}, \theta_{3} \in [-\pi/6, \pi/6], \quad \dot{\theta_{1}}, \dot{\theta_{2}}, \dot{\theta_{3}} \in [-1, 1].
\end{equation*}

Performance and efficiency of MAMPC, at various stages of training, is illustrated in Figure~\ref{fig:triplependulum_exp}.
The total and per-step running times of MAMPC are summarized in the second columns of Table~\ref{tab:sum-total-running-time} and Table~\ref{tab:sum-perstep-running-time}, respectively.

\begin{figure}[H]
    \centering
    \begin{tikzpicture}[font=\sffamily]
        \node[right] at (-4.5, 1.75) {\footnotesize MAMPC Initialized};
        \node[left] at (4.5, 1.75) {$\vert \mathcal{D} \vert = 0$};
        \draw[line width=0.25mm] (-4.5, 1.5) -- (4.5, 1.5);
        \node[inner sep=0pt] (time) at (0,0) {\includegraphics[width=.49\textwidth]{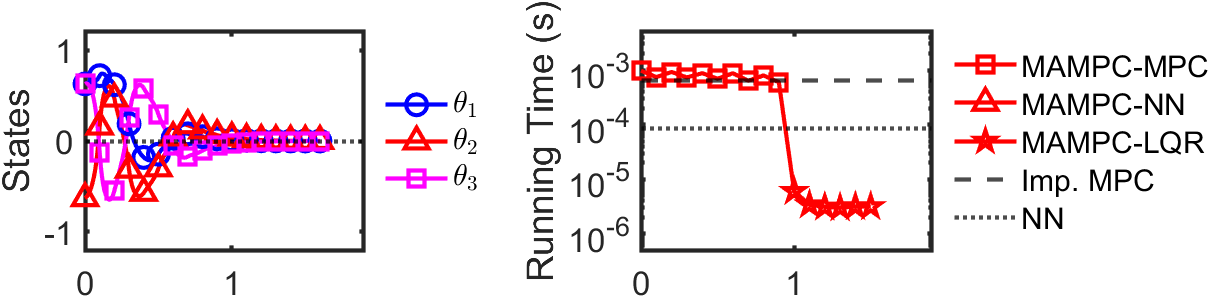}};

        \node[right] at (-4.5, -1.75) {\footnotesize MAMPC Being Trained};
        \node[left] at (4.5, -1.75) {$\vert \mathcal{D} \vert = 20000$};
        \draw[line width=0.25mm] (-4.5, -2) -- (4.5, -2);
        \node[inner sep=0pt] (time) at (0,-3.5) {\includegraphics[width=.49\textwidth]{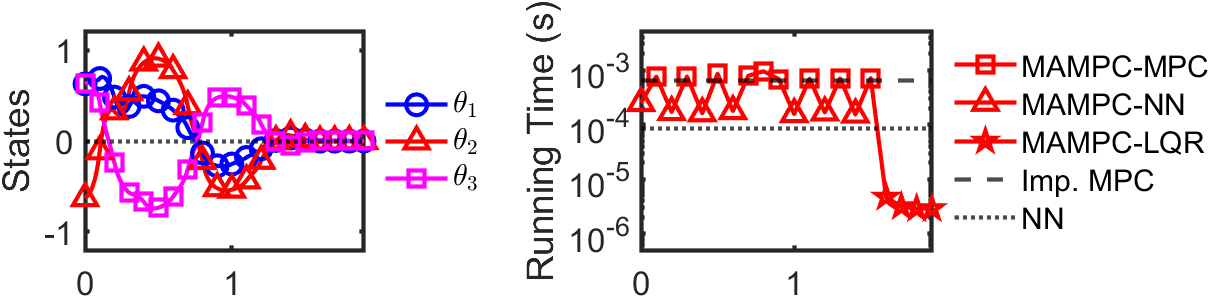}};

        \node[right] at (-4.5, -5.25) {\footnotesize MAMPC Converged};
        \node[left] at (4.5, -5.25) {$\vert \mathcal{D} \vert = 30000$};
        \draw[line width=0.25mm] (-4.5, -5.5) -- (4.5, -5.5);
        \node[inner sep=0pt] (time) at (0,-7) {\includegraphics[width=.49\textwidth]{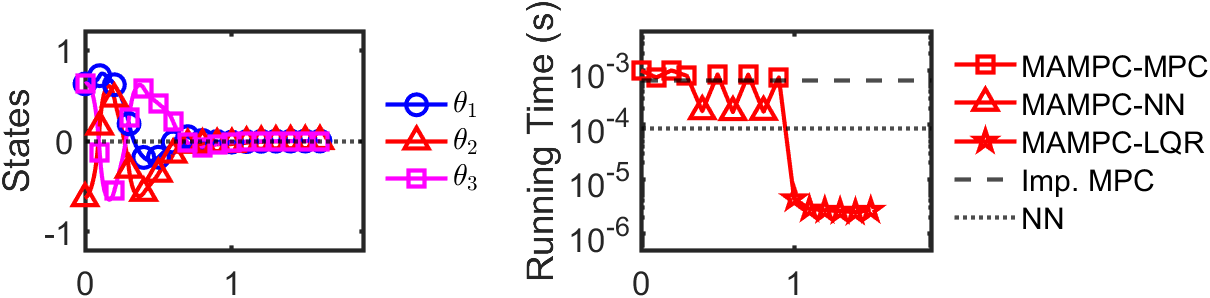}};
    \end{tikzpicture}
    \caption{Application of alternating-authority MAMPC to the triple pendulum problem.
    At initialization, the NN mode fails forward verification and thus never get invoked.
    As the NN is being trained, its likelihood of passing forward verification increases.
    At last, a well trained NN takes over most of MPC steps, resulting in a more efficient operation.}
    \label{fig:triplependulum_exp}
\end{figure}
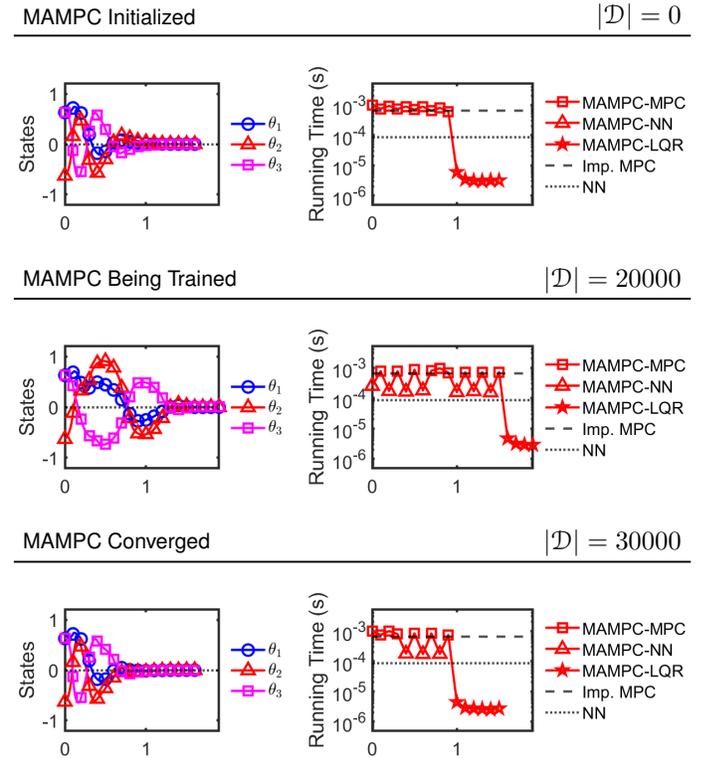

\subsection{Bicopter}

The third model is a bicopter as shown in Figure~\ref{fig:bicopter}.
The goal is to hover the bicopter in air.

The state is defined as $\vec{x} \coloneqq [x \quad \dot{x} \quad y \quad \dot{y} \quad \theta \quad \dot{\theta}]^{\top} \in \mathbb{R}^{4}\times([-\pi, \pi) \times \mathbb{R})$, where $x, y$ are horizontal and vertical translations and $\theta$ is the angle of tilting.
The input is defined as $\vec{u} \coloneqq [u_{1} \quad u_{2}]^{\top} \in \mathbb{R}^{2}$, where $u_{1}, u_{2}$ are the thrust exerted by the left and right propellers, respectively.
The control objective is to steer the bicopter to the origin by applying limited thrusts $u_{1}, u_{2} \in [0.1, 9.1572] \hspace{0.5em} \SI{}{\newton}$ at the two propellers.

\begin{figure}[H]
    \centering
    \includegraphics[width=0.305\textwidth]{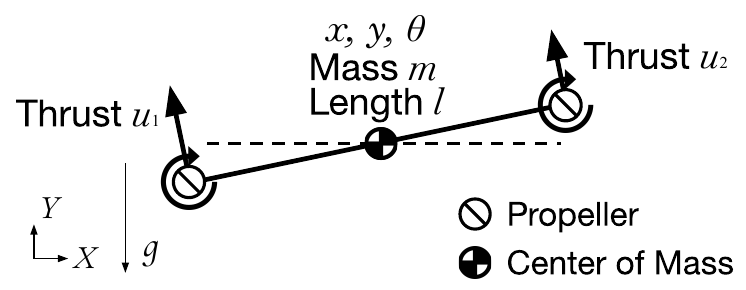}
    \caption{Diagram of a bicopter model.
    Each of the two propellers is capable of generating a thrust between $\SI{0.1}{\newton}$ and $\SI{9.1572}{\newton}$ to keep the copter hover in air.
    The thrust limits are chosen such that hovering thrusts are approximately in the center of the thrust limit.}
    \label{fig:bicopter}
\end{figure}

The MPC design highlights a sampling interval of $\SI{0.1}{\second}$, a planning horizon $N = 20$, and eight optimization variables per step, resulting in a LCQP of $\SI{2.0}{\second}$ prediction horizon and 160 variables.
To handle plant nonlinearity, we linearize the system around the equilibrium.

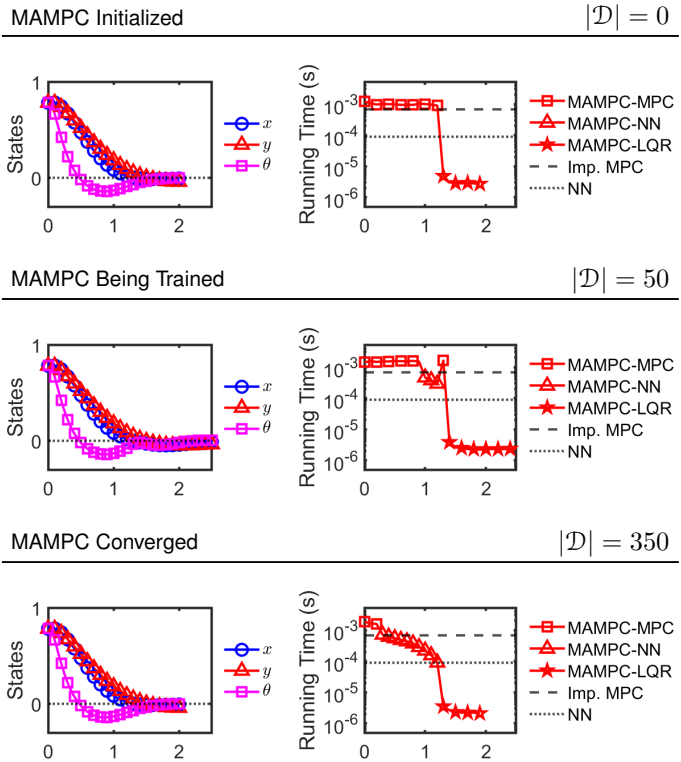
\begin{figure}[H]
    \centering
    \begin{tikzpicture}[font=\sffamily]
        \node[right] at (-4.5, 1.75) {\footnotesize MAMPC Initialized};
        \node[left] at (4.5, 1.75) {$\vert \mathcal{D} \vert = 0$};
        \draw[line width=0.25mm] (-4.5, 1.5) -- (4.5, 1.5);
        \node[inner sep=0pt] (time) at (0,0) {\includegraphics[width=.49\textwidth]{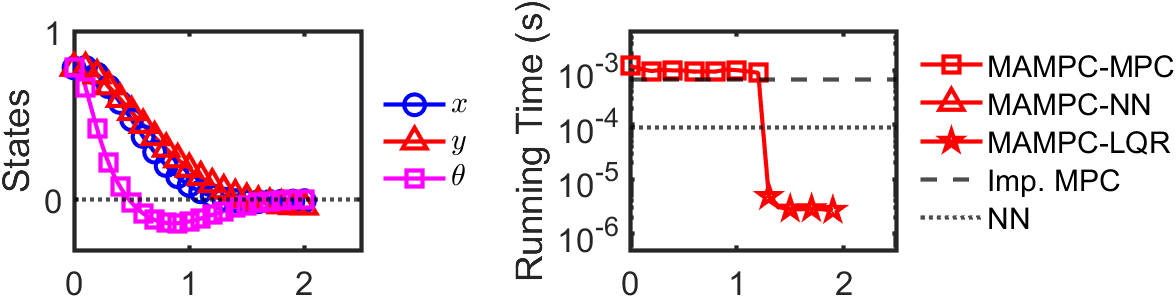}};

        \node[right] at (-4.5, -1.75) {\footnotesize MAMPC Being Trained};
        \node[left] at (4.5, -1.75) {$\vert \mathcal{D} \vert = 50$};
        \draw[line width=0.25mm] (-4.5, -2) -- (4.5, -2);
        \node[inner sep=0pt] (time) at (0,-3.5) {\includegraphics[width=.49\textwidth]{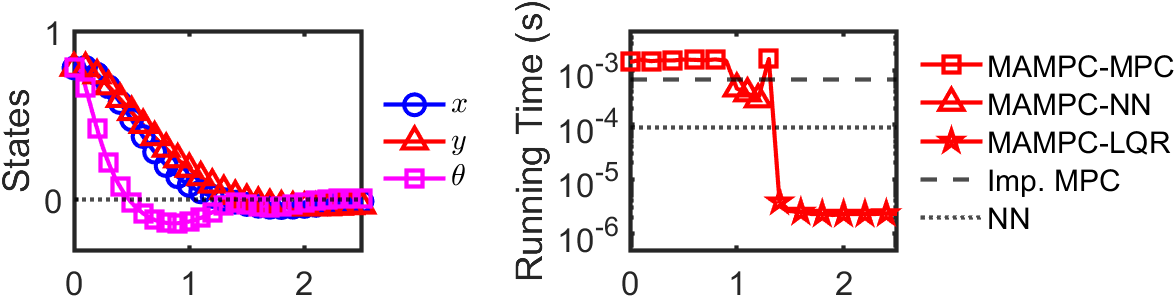}};

        \node[right] at (-4.5, -5.25) {\footnotesize MAMPC Converged};
        \node[left] at (4.5, -5.25) {$\vert \mathcal{D} \vert = 350$};
        \draw[line width=0.25mm] (-4.5, -5.5) -- (4.5, -5.5);
        \node[inner sep=0pt] (time) at (0,-7) {\includegraphics[width=.49\textwidth]{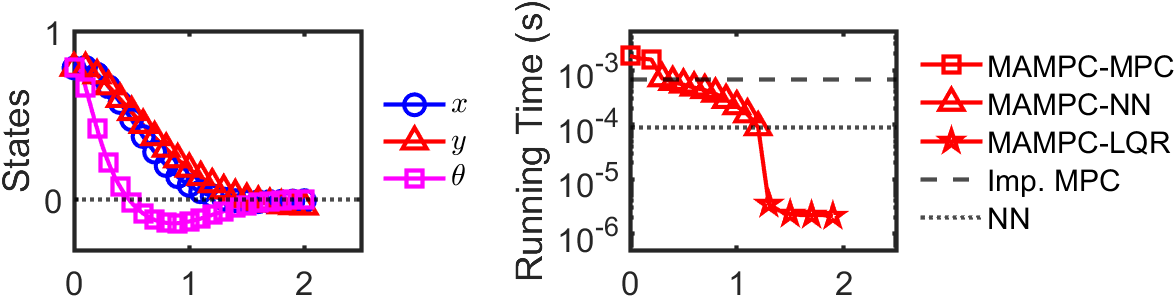}};
    \end{tikzpicture}
    \caption{Application of standard MAMPC to the bicopter problem.
    At initialization, the NN mode fails forward verification and thus never get invoked.
    As the NN is being trained, its likelihood of passing forward verification increases.
    At last, a well trained NN takes over most of MPC steps, resulting in a more efficient operation.}
    \label{fig:bicopter_exp}
\end{figure}

We apply a standard MAMPC to control the bicopter with 
\begin{equation*}
    \mathcal{R}_{\textnormal{LQR}} = \{\vec{x} \in \mathbb{R}^{6} \mid \Vert \vec{x} \Vert_{2} \leq 0.5\}, \quad N_{\textnormal{LQR}} = 10.
\end{equation*}
The NN is a three-layer, 50-neuron MLP trained through supervised learning with data uniformly randomly sampled from 
\begin{equation*}
    x, y, \theta \in [-\pi/2, \pi/2], \quad \dot{x}, \dot{y}, \dot{\theta} \in [-1, 1].
\end{equation*}

Performance and efficiency of MAMPC, at various stages of training, is illustrated in Figure~\ref{fig:bicopter_exp}.
The total and per-step running times of MAMPC are summarized in the third columns of Table~\ref{tab:sum-total-running-time} and Table~\ref{tab:sum-perstep-running-time}, respectively.

\subsection{Quadcopter}
As an extension to the ideal bicopter model, the last model is a quadcopter as shown in Figure~\ref{fig:quadcopter} \cite{alexis2012model}.
The goal is also to hover the copter in air but with more realistic system dynamics.

The state is defined as $\vec{x} \coloneqq [x \; \dot{x} \; y \; \dot{y} \; z \; \dot{z} \; \phi \; \dot{\phi} \; \theta \; \dot{\theta} \; \psi \; \dot{\psi}]^\top \in \mathbb{R}^{6}\times([-\pi, \pi) \times \mathbb{R})^{3}$, where $x, y, z$ represent translations and $\phi, \theta, \psi$ represent rotations corresponding to the three ZYX Euler angles roll, pitch, and yaw.
The input is $\vec{u} = [u_1 \quad u_2 \quad u_3 \quad u_4]^\top \in \mathbb{R}^{4}$, where $u_i$ is the rotational speeds of the $i^\text{th}$ propeller for $i = 1, 2, 3, 4$.
The control objective is to steer the quadcopter to the origin by applying bounded rotations $u_{1}, u_{2}, u_{3}, u_{4} \in [0, 313.96]$ in $\SI{}{\radian\per\second}$ at the four propellers.

\begin{figure}[H]
    \centering
    \includegraphics[width=0.35\textwidth]{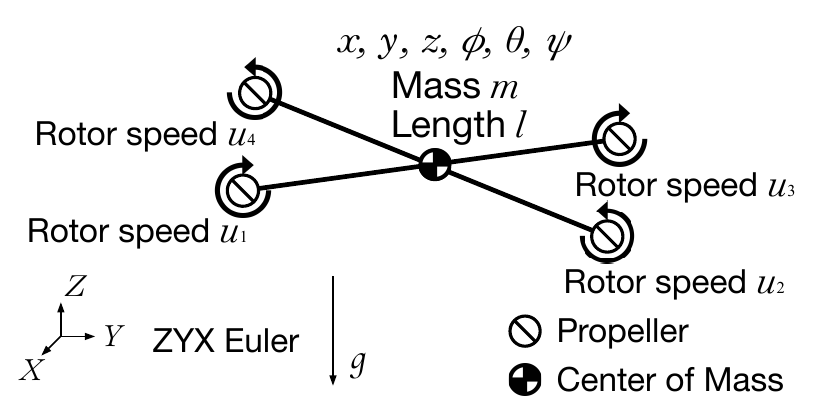}
    \caption{Diagram of a quadcopter model.
    Each of the four propellers is capable of rotating at speed between $\SI{0}{\radian\per\second}$ and $\SI{313.96}{\radian\per\second}$ to keep the copter hover in air.
    The rotor limits are chosen such that hovering rotor rotations are approximately in the center of the rotation range.}
    \label{fig:quadcopter}
\end{figure}

The MPC design highlights a sampling interval of $\SI{0.1}{\second}$, a planning horizon $N = 20$, and 16 optimization variables per step, resulting in a LCQP of $\SI{2.0}{\second}$ prediction horizon and 320 variables.
Nonlinearity is handled by linearization around the equilibrium.

Because the closed-loop quadcopter dynamics with MPC takes many steps to converge, we choose a way-point MAMPC with
\begin{equation*}
    \begin{aligned}
        &\mathcal{D}_{\textnormal{WP}} = \{\vec{x} \in \mathbb{R}^{12} \mid \Vert \vec{x} \Vert_{2} \leq 2\}, \quad N_{\textnormal{WP}} = 10,\\
        &\mathcal{R}_{\textnormal{LQR}} = \{\vec{x} \in \mathbb{R}^{12} \mid \Vert \vec{x} \Vert_{2} \leq 0.5\}, \quad N_{\textnormal{LQR}} = 10.
    \end{aligned}
\end{equation*}
The NN features a four-layer, 60-neuron MLP trained through supervised learning with data uniformly randomly sampled from 
\begin{equation*}
\begin{aligned}
    &x, y, z \in [-0.5, 0.5], \quad \dot{x}, \dot{y}, \dot{z} \in [-0.1, 0.1], \\
    &\phi, \theta \in [-\pi/6, \pi/6], \psi \in [-\pi/4, \pi/4], \dot{\phi}, \dot{\theta}, \dot{\psi} \in [-0.1, 0.1].
\end{aligned}
\end{equation*}

Performance and efficiency of MAMPC, at various stages of training, is illustrated in Figure~\ref{fig:quadcopter_exp}.
The total and per-step running times of MAMPC are summarized in the last columns of Table~\ref{tab:sum-total-running-time} and Table~\ref{tab:sum-perstep-running-time}, respectively.

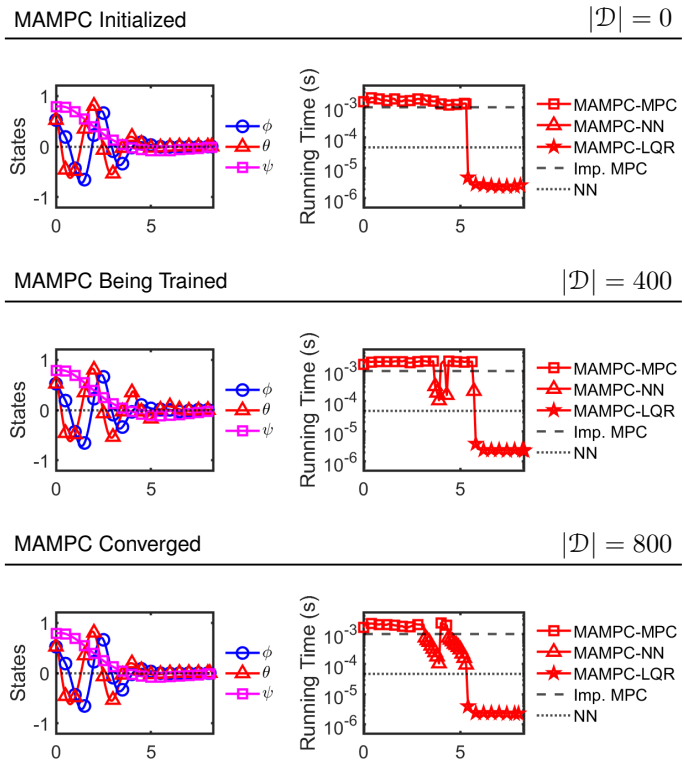
\begin{figure}[H]
    \centering
    \begin{tikzpicture}[font=\sffamily]
        \node[right] at (-4.5, 1.75) {\footnotesize MAMPC Initialized};
        \node[left] at (4.5, 1.75) {$\vert \mathcal{D} \vert = 0$};
        \draw[line width=0.25mm] (-4.5, 1.5) -- (4.5, 1.5);
        \node[inner sep=0pt] (time) at (0,0) {\includegraphics[width=.49\textwidth]{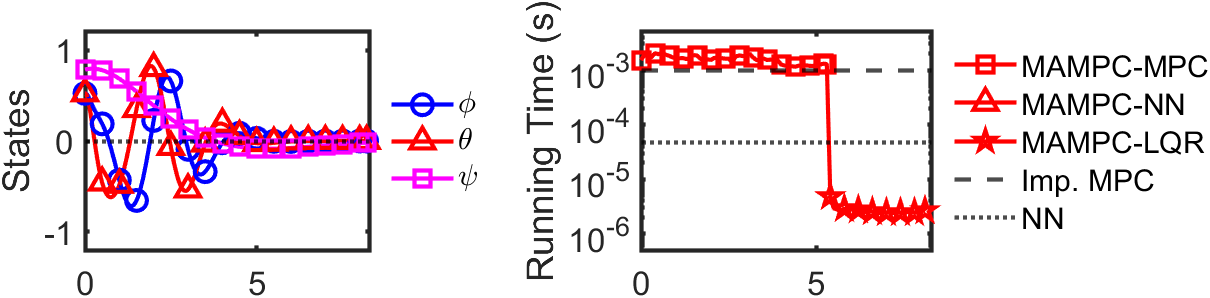}};

        \node[right] at (-4.5, -1.75) {\footnotesize MAMPC Being Trained};
        \node[left] at (4.5, -1.75) {$\vert \mathcal{D} \vert = 400$};
        \draw[line width=0.25mm] (-4.5, -2) -- (4.5, -2);
        \node[inner sep=0pt] (time) at (0,-3.5) {\includegraphics[width=.49\textwidth]{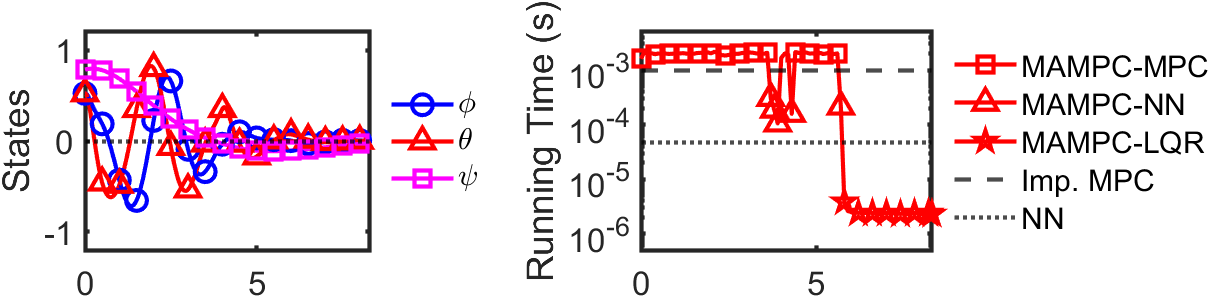}};

        \node[right] at (-4.5, -5.25) {\footnotesize MAMPC Converged};
        \node[left] at (4.5, -5.25) {$\vert \mathcal{D} \vert = 800$};
        \draw[line width=0.25mm] (-4.5, -5.5) -- (4.5, -5.5);
        \node[inner sep=0pt] (time) at (0,-7) {\includegraphics[width=.49\textwidth]{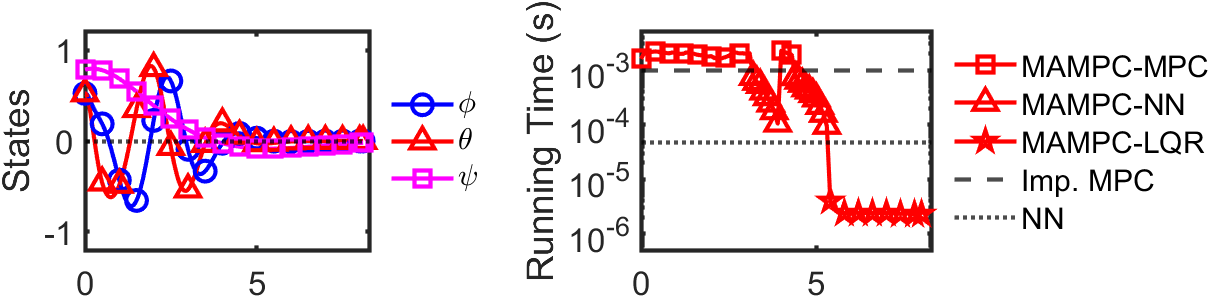}};
    \end{tikzpicture}
    \caption{Application of way-point MAMPC to the quadcopter problem.
    At initialization, the NN mode fails forward verification and thus never get invoked.
    As the NN is being trained, its likelihood of passing forward verification increases.
    At last, a well trained NN takes over most of MPC steps, resulting in a more efficient operation.}
    \label{fig:quadcopter_exp}
\end{figure}

\begin{table}[h]
  \small
  \centering
  \caption{Summary of total running time in millisecond}
  \resizebox{0.49\textwidth}{!}{\begin{tabular}{lcccc}
    \toprule
                & Pendulum & Tri. Pend. & Bicopter & Quadcopter\\
    \midrule
    Imp. MPC              & $10.17 \pm 3.41$     & $13.03 \pm 3.44$    & $15.58 \pm 3.12$     & $92.66 \pm 11.15$\\
    MAMPC (initialized)   & $7.37 \pm 2.40$      & $8.54 \pm 1.36$     & $14.69 \pm 3.97$     & $93.55 \pm 11.67$\\
    MAMPC (being train.)  & $7.21 \pm 2.11$      & $9.00 \pm 1.69$     & $20.04 \pm 4.74$     & $117.64 \pm 13.74$\\
    MAMPC (converged)     & $3.63 \pm 0.94$      & $7.11 \pm 1.63$     & $9.66 \pm 1.85$     & $87.84 \pm 16.69$\\
    \bottomrule
  \end{tabular}}
  \label{tab:sum-total-running-time}
\end{table}

\begin{table}[h]
  \small
  \centering
  \caption{Summary of per-step running time in mircosecond}
  \resizebox{0.49\textwidth}{!}{\begin{tabular}{lcccc}
    \toprule
                        & Pendulum             & Tri. Pend.        & Bicopter          & Quadcopter\\
    \midrule
    Imp. MPC                 & $686.39 \pm 248.03$     & $695.79 \pm 178.03$    & $743.84 \pm 152.67$     & $1032.76 \pm 89.52$\\
    MAMPC-MPC$^*$            & $1529.29 \pm 411.57$    & $897.02 \pm 226.43$    & $1777.19 \pm 404.51$   & $2100.90 \pm 401.96$\\
    MAMPC-NN$^*$             & $105.23 \pm 33.60$      & $187.53 \pm 51.62$     & $418.31 \pm 90.21$     & $420.92 \pm 95.02$\\
    MAMPC-LQR$^*$            & $1.86 \pm 0.46$         & $3.11 \pm 0.92$        & $2.42 \pm 0.55$        & $2.65 \pm 0.67$\\
    NN$^\dagger$             & $59.85 \pm 21.86$       & $79.86 \pm 24.20$      & $98.60 \pm 24.88$      & $56.25 \pm 23.89$\\
    \bottomrule
    \multicolumn{5}{l}{\footnotesize $*$ MAMPC-MPC, MAMPC-NN, and MAMPC-LQR are the MPC, NN, and LQR modes of the MAMPC, respectively.}\\
    \multicolumn{5}{l}{\footnotesize $\dagger$ NN is the NN mode of MAMPC \textit{without} forward verification.}\\
  \end{tabular}}%
  \label{tab:sum-perstep-running-time}
\end{table}

\Fangyu{Running time gradually decreases, a notion of acquired fluency. Draw analogy with human learning. This requires some numerical results.}

\section{Discussions}
\label{sec:discussions}

As shown in Figure~\ref{fig:pendulum_exp}, Figure~\ref{fig:triplependulum_exp}, Figure~\ref{fig:bicopter_exp}, and Figure~\ref{fig:quadcopter_exp} and in Table~\ref{tab:sum-total-running-time}, the MAMPC gains efficiency over time through learning on the trajectory data generated by the implicit MPC policy.
Upon initialization, the MAMPC policy, despite of not being the most efficient, is immediately functional.
After learning on a few samples, the NN mode behaves closer to the MPC mode but still lacks in accuracy.
This period of learning sometimes results an temporary degradation in control efficacy and an temporary increase in running time.
In particular, see Figure~\ref{fig:pendulum_exp} and Figure~\ref{fig:triplependulum_exp}.
After training on sufficient number of samples, the NN mode eventually converges to and takes over the implicit MPC.
Consequently, the overall control efficacy is recovered and computational efficiency improved.

After the NN mode converges, the resulting MAMPC is shown to have superior amortized running time but an inferior worst-case running time.
As the shown in the first four rows of Table~\ref{tab:sum-perstep-running-time}, the LQR mode and NN mode of MAMPC is faster than implicit MPC, while the MPC mode of MAMPC is slower than the implicit method due to overhead from computing the switch conditions.
Therefore, in the worst case, the MAMPC will take more time to stabilize the plant than the implicit MPC.
However, if the MAMPC spends sufficient number of steps in the NN and LQR modes, it will still has faster average running time, which is the case observed in the four numerical experiments.

If the target system is capable of parallel computation, it is possible the reduce the worst-time running time down to the baseline running time of the implicit MPC.
For example, a possible parallelization strategy for standard MAMPC is as follows: \textsl{1)} compute the three control modes of MAMPC in three \textit{parallel} threads, \textsl{2)} among those that pass their corresponding switch conditions, execute whichever that finishes first, and \textsl{3)} reset and repeat for next control step.
Ignoring communication overhead, one can show that the worst-case running time is roughly identical to that of the implicit MPC.
Note that under this setup the MAMPC scheme is \textit{at worst} equivalent to the dual-mode MPC in~\cite{michalska1993robust}.

Of course, it should be noted that there is no guarantee that the resulting MAMPC will be more efficient than the implicit MPC, because there is no guarantee that the NN will converge to an accurate surrogate policy.
Comparing to methods like~\cite{zhang2019safe}, where the implicit MPC controller is completely replaced by a NN controller, we trade a lack of deterministic guarantee on stability for a lack of deterministic guarantee on running time improvement.
We believe that such a trade-off is often desirable in practice: in industrial applications, safety almost always precedes efficiency.

Beyond the direct implications of the experiments, we make following remarks on set design, forward verification, additional running time optimizations, and limitations.

\subsubsection{Set Design}
Identification and design of feasibility set $\mathbb{X}_{0}$ and attraction set $\mathcal{R}_{\textnormal{LQR}}$ is very challenging. 
One one hand, it is easier to just approximate them by conservative set estimates for robustness measures. 
For example, one can replace $\mathbb{X}_{0}$ and $\mathcal{R}_{\textnormal{LQR}}$ with $\mathbb{X}_{0}^{-}$ and $\mathcal{R}_{\textnormal{LQR}}^{-}$, where $\mathbb{X}_{0}^{-} \subset \mathbb{X}_{0}$ and $\mathcal{R}_{\textnormal{LQR}}^{-} \subset \mathcal{R}_{\textnormal{LQR}}$.
On the other hand, however, to guarantee computational efficiency, one should make sets $\mathcal{R}_{\textnormal{LQR}}, \mathcal{R}_{\textnormal{NN}}$ as large as possible.
This way, faster modes of MAMPC get invoked more frequently than the slow, default mode of MAMPC, leading to a more efficient amortized running time performance.
Conservative set design makes the system more robust to model uncertainties but also slower in running time.
Balancing the trade-off between latency and robustness is therefore a design process.


\subsubsection{Forward verification}
The horizon used in forward verification immediately affects the maximum possible running time of the MAMPC.
A valid verification horizon should not make the running time of the MPC mode longer than the maximum allowable computational latency. 
Moreover, an additional rule of thumb is to choose a verification horizon such that the maximum possible running time of NN mode is on par with the running time of the implicit MPC.

The simulation method used for forward verification also directly impacts the worst-case latency of the MAMPC.
For continuous plants, we recommend to use forward Euler method with a sampling time that is as large as possible and to use a conservative set design to absorb integration errors.
Like designing sets, choosing the right sampling time is a balancing act between latency and robustness. 

\subsubsection{Additional Optimizations}
For small to moderately sized problems, we may consider developing a fail-free MAMPC as specified in Theorem~\ref{thm:2mampc-nofs}, since it will further speed up computation.
Fail-free MAMPC prevents controller from ever having to default to fail-safe mode and therefore skips the routine computation of forward verification in NN mode.
It significantly improves the amortized and worst-case running time.

Another possibility for reducing running time is to pre-compute the results of forward verification for the entire state space \textit{offline} and look up the results \textit{online}.
This look-up table approach is in spirit similar to explicit MPC.
The complexity of this approach also grows exponentially with state dimension, so it suffers the similar scalability issue like explicit MPC does.

Of course, a third possibility for reducing running time is to speed up computation of the trained neural network through techniques such as quantization, pruning, or model distillation.

Finally, it is sometimes useful to further improve the performance of the NN mode via unsupervised reinforcement learning.
This is usually the case when the default implicit MPC solver is sub-optimal, typically due to lack of convexity in the objective or constraints or due to large problem dimension induced by a long prediction horizon.

\section{Conclusion}
\label{sec:conclusion}

To improve \textit{amortized} computational efficiency of MPC, we have developed a triple-mode hybrid control named MAMPC.
Stability of MAMPC is guaranteed via forward simulation, while efficiency is achieved by replacing MPC with a more efficient NN or LQR, whenever stability permits.
Numerical experiments indicate that MAMPC often has a better \textit{amortized} running time but a slightly prolonged \textit{worst-case} per-step running time compared to implicit MPC method.

\bibliographystyle{unsrt}
\bibliography{references}

\begin{thebibliography}{10}

\bibitem{rawlings2009model}
James~B. Rawlings, David~Q. Mayne, and Moritz~M. Diehl.
\newblock {\em Model {P}redictive {C}ontrol: {T}heory, {C}omputation and
  {D}esign}.
\newblock Nob Hill Publishing, 2017.

\bibitem{borrelli2017predictive}
Francesco Borrelli, Alberto Bemporad, and Manfred Morari.
\newblock {\em Predictive {C}ontrol for {L}inear and {H}ybrid {S}ystems}.
\newblock Cambridge University Press, 2017.

\bibitem{borrelli2005mpc}
Francesco Borrelli, Paolo Falcone, Tamas Keviczky, Jahan Asgari, and Davor
  Hrovat.
\newblock {MPC}-based approach to active steering for autonomous vehicle
  systems.
\newblock {\em International Journal of Vehicle Autonomous Systems},
  3(2-4):265--291, 2005.

\bibitem{alexis2012model}
Kostas Alexis, George Nikolakopoulos, and Anthony Tzes.
\newblock Model predictive quadrotor control: attitude, altitude and position
  experimental studies.
\newblock {\em IET Control Theory and Applications}, 6(12):1812--1827, 2012.

\bibitem{kuindersma2016optimization}
Scott Kuindersma, Robin Deits, Maurice Fallon, Andr{\'e}s Valenzuela, Hongkai
  Dai, Frank Permenter, Twan Koolen, Pat Marion, and Russ Tedrake.
\newblock Optimization-based locomotion planning, estimation, and control
  design for the atlas humanoid robot.
\newblock {\em Autonomous robots}, 40(3):429--455, 2016.

\bibitem{bledt2018cheetah}
Gerardo Bledt, Matthew~J. Powell, Benjamin Katz, Jared Di~Carlo, Patrick~M.
  Wensing, and Sangbae Kim.
\newblock {MIT} {C}heetah 3: {D}esign and {C}ontrol of a {R}obust, {D}ynamic
  {Q}uadruped {R}obot.
\newblock In {\em 2018 IEEE/RSJ International Conference on Intelligent Robots
  and Systems (IROS)}, pages 2245--2252. IEEE, 2018.

\bibitem{michalska1993robust}
Hanna Michalska and David~Q. Mayne.
\newblock Robust {R}eceding {H}orizon {C}ontrol of {C}onstrained {N}onlinear
  {S}ystems.
\newblock {\em IEEE transactions on automatic control}, 38(11):1623--1633,
  1993.

\bibitem{diehl2005real}
Moritz Diehl, Hans~Georg Bock, and Johannes~P Schl{\"o}der.
\newblock {A} {R}eal-{T}ime {I}teration {S}cheme for {N}onlinear {O}ptimization
  in {O}ptimal {F}eedback {C}ontrol.
\newblock {\em SIAM Journal on Control and Optimization}, 43(5):1714--1736,
  2005.

\bibitem{mastalli2020crocoddyl}
Carlos Mastalli, Rohan Budhiraja, Wolfgang Merkt, Guilhem Saurel, Bilal
  Hammoud, Maximilien Naveau, Justin Carpentier, Ludovic Righetti, Sethu
  Vijayakumar, and Nicolas Mansard.
\newblock {C}rocoddyl: {A}n {E}fficient and {V}ersatile {F}ramework for
  {M}ulti-{C}ontact {O}ptimal {C}ontrol.
\newblock In {\em IEEE International Conference on Robotics and Automation},
  pages 2536--2542. IEEE, 2020.

\bibitem{houska2011auto}
Boris Houska, Hans~Joachim Ferreau, and Moritz Diehl.
\newblock {A}n {A}uto-{G}enerated {R}eal-{T}ime {I}teration {A}lgorithm for
  {N}onlinear {MPC} in the {M}icrosecond {R}ange.
\newblock {\em Automatica}, 47(10):2279--2285, 2011.

\bibitem{mattingley2012cvxgen}
Jacob Mattingley and Stephen Boyd.
\newblock {CVXGEN}: {a} code generator for embedded convex optimization.
\newblock {\em Optimization and Engineering}, 13(1):1--27, 2012.

\bibitem{faulwasser2016implementation}
Timm Faulwasser, Tobias Weber, Pablo Zometa, and Rolf Findeisen.
\newblock {I}mplementation of {N}onlinear {M}odel {P}redictive
  {P}ath-{F}ollowing {C}ontrol for an {I}ndustrial {R}obot.
\newblock {\em IEEE Transactions on Control Systems Technology},
  25(4):1505--1511, 2016.

\bibitem{kleff2021high}
S{\'e}bastien Kleff, Avadesh Meduri, Rohan Budhiraja, Nicolas Mansard, and
  Ludovic Righetti.
\newblock {H}igh-{F}requency {N}onlinear {M}odel {P}redictive {C}ontrol of a
  {M}anipulator.
\newblock In {\em IEEE International Conference on Robotics and Automation},
  2021.

\bibitem{summers2011multiresolution}
Sean Summers, Colin~N. Jones, John Lygeros, and Manfred Morari.
\newblock A {M}ultiresolution {A}pproximation {M}ethod for {F}ast {E}xplicit
  {M}odel {P}redictive {C}ontrol.
\newblock {\em IEEE Transactions on Automatic Control}, 56(11):2530--2541,
  2011.

\bibitem{kvasnica2013complexity}
Michal Kvasnica, Juraj Hled{\'\i}k, Ivana Rauov{\'a}, and Miroslav Fikar.
\newblock Complexity reduction of explicit model predictive control via
  separation.
\newblock {\em Automatica}, 49(6):1776--1781, 2013.

\bibitem{parisini1995receding}
Thomas Parisini and Riccardo Zoppoli.
\newblock A {R}eceding-horizon {R}egulator for {N}onlinear {S}ystems and a
  {N}eural {A}pproximation.
\newblock {\em Automatica}, 31(10):1443--1451, 1995.

\bibitem{parisini1998nonlinear}
Thomas Parisini, M~Sanguineti, and R~Zoppoli.
\newblock {N}onlinear {S}tabilization by {R}eceding-{H}orizon {N}eural
  {R}egulators.
\newblock {\em International Journal of Control}, 70(3):341--362, 1998.

\bibitem{zhang2019safe}
Xiaojing Zhang, Monimoy Bujarbaruah, and Francesco Borrelli.
\newblock Safe and {N}ear-{O}ptimal {P}olicy {L}earning for {M}odel
  {P}redictive {C}ontrol using {P}rimal-{D}ual {N}eural {N}etworks.
\newblock In {\em 2019 American Control Conference}, pages 354--359. IEEE,
  2019.

\bibitem{paulson2020approximate}
Joel~A Paulson and Ali Mesbah.
\newblock {A}pproximate {C}losed-{L}oop {R}obust {M}odel {P}redictive {C}ontrol
  with {G}uaranteed {S}tability and {C}onstraint {S}atisfaction.
\newblock {\em IEEE Control Systems Letters}, 4(3):719--724, 2020.

\bibitem{karg2020stability}
Benjamin Karg and Sergio Lucia.
\newblock Stability and feasibility of neural network-based controllers via
  output range analysis.
\newblock In {\em IEEE Conference on Decision and Control}, pages 4947--4954.
  IEEE, 2020.

\bibitem{chen2022large}
Steven~W Chen, Tianyu Wang, Nikolay Atanasov, Vijay Kumar, and Manfred Morari.
\newblock {L}arge {S}cale {M}odel {P}redictive {C}ontrol with {N}eural
  {N}etworks and {P}rimal {A}ctive {S}ets.
\newblock {\em Automatica}, 135:109947, 2022.

\bibitem{zoppoli2020neural}
Riccardo Zoppoli, Marcello Sanguineti, Giorgio Gnecco, and Thomas Parisini.
\newblock {\em {N}eural {A}pproximations for {O}ptimal {C}ontrol and
  {D}ecision}.
\newblock Springer, 2020.

\bibitem{nubert2020safe}
Julian Nubert, Johannes K{\"o}hler, Vincent Berenz, Frank Allg{\"o}wer, and
  Sebastian Trimpe.
\newblock {S}afe and {F}ast {T}racking on a {R}obot {M}anipulator: {R}obust
  {MPC} and {N}eural {N}etwork {C}ontrol.
\newblock {\em IEEE Robotics and Automation Letters}, 5(2):3050--3057, 2020.

\bibitem{DBLP:journals/corr/KingmaB14}
Diederik~P. Kingma and Jimmy Ba.
\newblock Adam: {A} {M}ethod for {S}tochastic {O}ptimization.
\newblock In {\em International Conference on Learning Representations}, 2015.

\bibitem{anderson2018optimal}
Brian D.~O. Anderson and John~B Moore.
\newblock {\em Optimal {C}ontrol: {L}inear {Q}uadratic {M}ethods}, pages
  53--54.
\newblock Dover Publications, 2018.

\bibitem{khalil2002nonlinear}
Hassan~K. Khalil.
\newblock {\em Nonlinear {S}ystems}.
\newblock Prentice Hall, 3rd edition, 2002.

\bibitem{sutton2018reinforcement}
Richard~S Sutton and Andrew~G Barto.
\newblock {\em Reinforcement {L}earning: {A}n {I}ntroduction}.
\newblock MIT press, 2018.

\bibitem{richards2018lyapunov}
Spencer~M Richards, Felix Berkenkamp, and Andreas Krause.
\newblock The {L}yapunov {N}eural {N}etwork: {A}daptive {S}tability
  {C}ertification for {S}afe {L}earning of {D}ynamical {S}ystems.
\newblock In {\em Conference on Robot Learning}, pages 466--476. PMLR, 2018.

\bibitem{chen2021learning}
Shaoru Chen, Mahyar Fazlyab, Manfred Morari, George~J Pappas, and Victor~M
  Preciado.
\newblock {L}earning {L}yapunov functions for hybrid systems.
\newblock In {\em Proceedings of the 24th International Conference on Hybrid
  Systems: Computation and Control}, pages 1--11, 2021.

\bibitem{wu2021mampcgithub}
Fangyu Wu.
\newblock Memory-{A}ugmented {M}odel {P}redictive {C}ontrol: {N}umerical
  {E}xperiments.
\newblock https://github.com/fywu85/mampc, 2021.

\end{thebibliography}

\appendix

\subsection{Proof of Theorem~\ref{thm:mampc-stability}}
\label{app:mampc-proof}
\begin{proof}
    Existence of MAMPC depends on existence of LQR and NN.
    A LQR controller of the form~\eqref{eq:lqr} can always be derived from~\eqref{eq:mpc} because  
    \textsl{1)} $f$ is continuously differentiable, so it can always be linearized to produce the $\mat{A}, \mat{B}$ matrices in~\eqref{eq:lqr};
    \textsl{2)} the linearized system $(\mat{A}, \mat{B})$ is by assumption \textit{stabilizable}, so region of attraction of LQR is nonempty, i.e., $\mathcal{R}_{\textnormal{LQR}} \neq \emptyset$;
    \textsl{3)} as for $\mat{Q}, \mat{R}$ we only require $\mat{Q} \succeq 0, \mat{R} \succ 0$, which is always possible;
    \textsl{4)} lastly, optimization problem~\eqref{eq:lqr} is just a relaxation of problem~\eqref{eq:mpc}, which can always be solved.

    Meanwhile, a NN controller trivially exists because we do not require NN to satisfy any properties other than the basic definition of a NN.
    This completes the first part of the proof on existence.
    
    
    To prove local asymptotic stability of standard MAMPC in $\mathbb{X}_{0}$, we partition $\mathbb{X}_{0}$ into three regions as follows
    \begin{equation*}
        \mathbb{X}_{0} = \mathcal{R}_{\textnormal{LQR}} \cup \mathcal{R}_{\textnormal{NN}} \cup \mathcal{R}_{\textnormal{MPC}},
    \end{equation*}
    where $\mathcal{R}_{\textnormal{NN}}$ is the set of states where the NN is invoked and $\mathcal{R}_{\textnormal{MPC}}$ is the rest of states in $\mathbb{X}_{0}$, i.e., $\mathcal{R}_{\textnormal{MPC}} \coloneqq \mathbb{X}_{0} \setminus (\mathcal{R}_{\textnormal{LQR}} \cup \mathcal{R}_{\textnormal{NN}})$.
    By construction, $\mathcal{R}_{\textnormal{LQR}}$, $\mathcal{R}_{\textnormal{NN}}$, and $\mathcal{R}_{\textnormal{MPC}}$ are mutually disjoint.
    
    For every state $\vec{x}[0] \in \mathcal{R}_{\textnormal{LQR}}$, $u_{\textnormal{LQR}}$ is invoked and the closed-loop system~\eqref{eq:clsys-mampc} will be equivalent to~\eqref{eq:clsys-lqr}, which is locally asymptotically stable.
    
    For every state $\vec{x}[0] \in \mathcal{R}_{\textnormal{NN}}$, $u_{\textnormal{NN}}$ is invoked and the closed-loop system~\eqref{eq:clsys-mampc} will be equivalent to~\eqref{eq:clsys-nn} until the state enters $\mathcal{R}_{\textnormal{LQR}}$ in no more than $N_{\textnormal{LQR}}$ steps. 
    The system will stay in $\mathcal{R}_{\textnormal{LQR}}$ and will be taken to the origin by $u_{\textnormal{LQR}}$ as $i \rightarrow \infty$.
    Consequently, the closed-loop system~\eqref{eq:clsys-mampc} is locally asymptotically stable in $\mathcal{R}_{\textnormal{MPC}} \cup \mathcal{R}_{\textnormal{NN}}$.
    
    For every state $\vec{x}[0] \in \mathcal{R}_{\textnormal{MPC}}$, $u_{\textnormal{MPC}}$ is invoked and the closed-loop system~\eqref{eq:clsys-mampc} will be equivalent to~\eqref{eq:clsys-mpc} for at least one time step.
    In the next time step, if the state enters $\mathcal{R}_{\textnormal{LQR}} \cup \mathcal{R}_{\textnormal{NN}}$, then the closed-loop system will converge to the origin as it is locally asymptotically stable in $\mathcal{R}_{\textnormal{LQR}} \cup \mathcal{R}_{\textnormal{NN}}$; otherwise, $\mathcal{R}_{\textnormal{MPC}}$ will be invoked again.
    Because the closed-loop system~\eqref{eq:clsys-mpc} is locally asymptotically stable, even if MPC never brings the state into $\mathcal{R}_{\textnormal{NN}}$, it will eventually steer the system into $\mathcal{R}_{\textnormal{LQR}}$ in finite time, which will then be taken to the origin by $u_{\textnormal{LQR}}$.
    Consequently, the closed-loop system~\eqref{eq:clsys-mampc} is locally asymptotically stable in $\mathbb{X}_{0}$.
\end{proof}

\subsection{Proof of Theorem~\ref{thm:aa-mampc-stability}}
\label{app:aa-mampc-proof}

\begin{proof}
    The first part of the proof on existence is identical to the one in Theorem~\ref{thm:mampc-stability}.
    
    To prove stability of alternating-authority MAMPC, we partition  $\mathbb{X}_{0}$ into two disjoint regions  $\mathbb{X}_{0} = \mathcal{R}_{\textnormal{LQR}} \cup \mathcal{R}_{\textnormal{AA}}$,
    where $\mathcal{R}_{\textnormal{AA}} \coloneqq \mathbb{X}_{0} \setminus \mathcal{R}_{\textnormal{LQR}}$.
    
    For every state $\vec{x}[0] \in \mathcal{R}_{\textnormal{LQR}}$, the closed-loop system~\eqref{eq:clsys-aa-mampc} is identical to closed-loop system~\eqref{eq:clsys-mampc}, which is locally asymptotically stable, as shown in Theorem~\eqref{thm:mampc-stability}.
    
    Next, we prove by contradiction that the system will never escape $\mathbb{X}_{0}$.
    Suppose there exits an escaping control $\bar{\vec{u}}$ such that $\bar{\vec{x}} \in \mathbb{X}_{0}$ but $f(\bar{\vec{x}}, \bar{\vec{u}}) \notin \mathbb{X}_{0}$. 
    This escaping control $\bar{\vec{u}}$ must be produced by either $u_{\textnormal{MPC}}$, $u_{\textnormal{NN}}$, or $u_{\textnormal{LQR}}$, that is, $\bar{\vec{u}} = u_{\textnormal{MPC}}(\bar{\vec{x}})$, $\bar{\vec{u}} = u_{\textnormal{NN}}(\bar{\vec{x}})$, or $\bar{\vec{u}} = u_{\textnormal{LQR}}(\bar{\vec{x}})$.
    If $\bar{\vec{u}} = u_{\textnormal{MPC}}(\bar{\vec{x}})$, then $u_{\textnormal{MPC}}$ is not persistently feasible in $\mathbb{X}_{0}$, which is not possible because we assume that $u_{\textnormal{MPC}}$ satisfies Theorem~\eqref{thm:as-mpc}.
    If $\bar{\vec{u}} = u_{\textnormal{NN}}(\bar{\vec{x}})$, then the following switch condition must hold: $f(\bar{\vec{x}}, \bar{\vec{u}}) \in \mathbb{X} \subseteq \mathbb{X}_{0}$.
    However, this is not possible since $\bar{\vec{u}}$ is assumed to be an escaping control, that is, $f(\bar{\vec{x}}, \bar{\vec{u}}) \notin \mathbb{X}_{0}$.
    If $\bar{\vec{u}} = u_{\textnormal{LQR}}(\bar{\vec{x}})$, then there must be $\bar{\vec{x}} \in \mathcal{R}_{\textnormal{LQR}}$.
    Since $\mathcal{R}_{\textnormal{LQR}}$ is positively invariant, we must have $f(\bar{\vec{x}}, \bar{\vec{u}}) \in \mathcal{R}_{\textnormal{LQR}} \subset \mathbb{X}_{0}$, which is not possible since $f(\bar{\vec{x}}, \bar{\vec{u}}) \notin \mathbb{X}_{0}$.
    Therefore, the system will never escape $\mathbb{X}_{0}$.
\end{proof}

\subsection{Proof of Theorem~\ref{thm:wp-mampc-stability}}
\label{app:wp-mampc-proof}

\begin{proof}
    The first part of the proof on existence is identical to the one in Theorem~\ref{thm:mampc-stability}.
    
    To prove stability of way-point MAMPC, we partition  $\mathbb{X}_{0}$ into five disjoint regions as follows
    \begin{equation*}
        \mathbb{X}_{0} = \mathcal{R}_{\textnormal{LQR}} \cup \mathcal{R}_{\textnormal{NN}}^{i} \cup \mathcal{R}_{\textnormal{MPC}}^{i} \cup \mathcal{R}_{\textnormal{NN}}^{o} \cup \mathcal{R}_{\textnormal{MPC}}^{o},
    \end{equation*}
    where $\mathcal{R}_{\textnormal{NN}}^{i}$ is the set of states in or on the boundary of $\mathcal{D}_{\textnormal{WP}}$ where the NN is invoked;
    $\mathcal{R}_{\textnormal{NN}}^{o}$ is the set of states in $\mathbb{X}_{0} \setminus \mathcal{D}_{\textnormal{WP}}$ where the NN is invoked;
    $\mathcal{R}_{\textnormal{MPC}}^{i} \coloneqq \mathcal{D}_{\textnormal{WP}} \setminus (\mathcal{R}_{\textnormal{LQR}} \cup \mathcal{R}_{\textnormal{NN}}^{i})$; and
    $\mathcal{R}_{\textnormal{MPC}}^{o} \coloneqq (\mathbb{X}_{0} \setminus \mathcal{D}_{\textnormal{WP}}) \setminus \mathcal{R}_{\textnormal{NN}}^{o}$.

    For every state $\vec{x} \in \mathcal{R}_{\textnormal{LQR}}$, the closed-loop system~\eqref{eq:clsys-wp-mampc} is identical to closed-loop system~\eqref{eq:clsys-mampc}, which is locally asymptotically stable, as shown in Theorem~\eqref{thm:mampc-stability}.
    
    Next, we prove that the system will eventually enter and stay within $\mathcal{D}_{\textnormal{WP}}^{+}$.
    
    For every state $\vec{x}[0] \in \mathbb{X}_{0} \setminus \mathcal{D}_{\textnormal{WP}}$, either NN or MPC will steer the system into $\mathcal{D}_{\textnormal{WP}}^{+}$ in finite time.
    If NN is ever invoked, the system~\eqref{eq:clsys-wp-mampc} will be taken into $\mathcal{D}_{\textnormal{WP}}$ in no more than $N_{\textnormal{WP}}$ steps.
    Otherwise, MPC will bring the system~\eqref{eq:clsys-wp-mampc} into $\mathcal{D}_{\textnormal{WP}}$ in finite time since the MPC is locally asymptotically stable in $\mathbb{X}_{0}$.

    For every state $\vec{x}[0] \in \mathcal{D}_{\textnormal{WP}}$, we prove by contradiction that there exists a $j \geq 0$ such that for all $i \geq j, \vec{x}[i] \in \mathcal{D}_{\textnormal{WP}}^{+}$.
    Suppose there exists a $\vec{x}[0] = \bar{\vec{x}} \in \mathcal{D}_{\textnormal{WP}}$ such that for every $k \geq 0$, there exists an $i \geq k$ such that $\vec{x}[i] \notin \mathcal{D}_{\textnormal{WP}}^{+}$.
    Without loss of generality, define a control $\bar{\vec{u}}$ such that $\vec{x}[i-1] \in \mathcal{D}_{\textnormal{WP}}^{+}$ but $\vec{x}[i] = f(\vec{x}[i-1], \bar{\vec{u}}) \notin \mathcal{D}_{\textnormal{WP}}^{+}$. 

    This escaping control $\bar{\vec{u}}$ must be produced by either $u_{\textnormal{MPC}}$, $u_{\textnormal{NN}}$, or $u_{\textnormal{LQR}}$, that is, $\bar{\vec{u}} = u_{\textnormal{MPC}}(\vec{x}[i-1])$, $\bar{\vec{u}} = u_{\textnormal{NN}}(\vec{x}[i-1])$, or $\bar{\vec{u}} = u_{\textnormal{LQR}}(\vec{x}[i-1])$.
    If $\bar{\vec{u}} = u_{\textnormal{MPC}}(\vec{x}[i-1])$, then the close-loop system~\eqref{eq:clsys-wp-mampc} is equivalent to system~\eqref{eq:clsys-mpc}.
    Because $\vec{x}[i] \notin \mathcal{D}_{\textnormal{WP}}^{+}$, we have
    \begin{equation*}
        V(\vec{x}[i]) = V(f(\vec{x}[i-1], u_{\textnormal{MPC}}(\vec{x}[i-1])) > V(\vec{x}[i-1]),
    \end{equation*}
    but this is not possible because $V$ is a Lyapunov function of the system~\eqref{eq:clsys-mpc}, i.e., 
    \begin{equation*}
        V(f(\vec{x}[i-1], u_{\textnormal{MPC}}(\vec{x}[i-1])) \leq V(\vec{x}[i-1]).
    \end{equation*}
    If $\bar{\vec{u}} = u_{\textnormal{NN}}(\vec{x}[i-1])$, then $f(\vec{x}[i-1], u_{\textnormal{NN}}(\vec{x}[i-1])) \notin \mathcal{D}_{\textnormal{WP}}$, which is not possible because invocation of NN implies that $f(\vec{x}[i-1], u_{\textnormal{NN}}(\vec{x}[i-1])) \in \mathcal{D}_{\textnormal{WP}}$.
    If $\bar{\vec{u}} = u_{\textnormal{LQR}}(\vec{x}[i-1])$, then by assumption $\lim_{i\rightarrow \infty}\vec{x}[i]$ does not exist.
    However, this is impossible because $\lim_{i\rightarrow \infty}\vec{x}[i] = \vec{0}$ by the local asymptotic stability of the system~\eqref{eq:clsys-wp-mampc} in $\mathcal{R}_{\textnormal{LQR}}$.
    Consequently, the assumption must be false.
\end{proof}

\end{document}